\documentclass[12pt]{amsart}
\pdfoutput=1
\usepackage[a4paper,margin=2cm]{geometry}

\usepackage[T2A]{fontenc}
\usepackage[utf8]{inputenc}

\usepackage{amsmath,amssymb,amsthm,mathtools,enumerate,extarrows,stmaryrd}

\usepackage[oldstyle]{libertine}

\usepackage[dvipsnames]{xcolor}
\usepackage{hyperref}
\hypersetup{colorlinks=true,linkcolor=blue,citecolor=teal,filecolor=magenta,urlcolor=cyan}
\usepackage{mathtools}
\usepackage{extarrows}

\usepackage{url}
\usepackage[backend=bibtex,style=alphabetic,sorting=nyt,isbn=false,url=false,doi=true,maxalphanames=10,minalphanames=4,mincitenames=4,maxcitenames=10,minnames=4,maxnames=10,giveninits=true,maxbibnames=99]{biblatex}
\theoremstyle{plain}
\newtheorem{theorem}{Theorem}[section]

\newtheorem{corollary}[theorem]{Corollary}
\newtheorem{conjecture}[theorem]{Conjecture}
\newtheorem{lemma}[theorem]{Lemma}

\theoremstyle{definition}
\newtheorem{definition}[theorem]{Definition}
\newtheorem{example}[theorem]{Example}
\newtheorem{notation}[theorem]{Notation}
\theoremstyle{remark}
\newtheorem{remark}[theorem]{Remark}

\newcommand{\Vev}[1]{\big\langle 0 \big|\, {#1} \,\big| 0 \big\rangle}
\newcommand{\VEV}[1]{{\big\langle 0 \big| {#1} \big| 0 \big\rangle}}
\newcommand{\vac}{| 0 \rangle}
\newcommand{\covac}{\langle 0 |}

\newcommand{\Map}{\mathsf{Map}}
\newcommand{\FSMap}{\mathsf{FSMap}}

\DeclareMathOperator{\Id}{Id}

\newcommand{\cA}{\mathcal{A}}
\newcommand{\cD}{\mathcal{D}}
\newcommand{\bJ}{\mathbb{J}}
\newcommand{\cS}{\mathcal{S}}
\newcommand{\cW}{\mathcal{W}}
\newcommand{\cF}{\mathcal{F}}
\renewcommand{\C}{\mathbb{C}}
\newcommand{\wgt}{\mathrm{wgt}}
\newcommand{\Z}{\mathbb{Z}}
\newcommand{\bt}{\boldsymbol{t}}

\newcommand{\cE}{\mathcal{E}}

\newcommand{\lwgt}{\mathfrak{W}}
\newcommand{\lwgtpart}{\mathfrak{w}}

\newcommand{\Eop}{\mathfrak{E}}
\newcommand{\genwgt}{\mathrm{weight}}

\newcommand{\Mf}{\mathsf{M}}

\begin{document}
	
\title[Generalised ordinary vs fully simple duality]{Generalised ordinary vs fully simple duality for $n$-point functions and a proof of the Borot--Garcia-Failde conjecture}

\author[B.~Bychkov]{Boris~Bychkov}
\address{B.~B.: Faculty of Mathematics, National Research University Higher School of Economics, Usacheva 6, 119048 Moscow, Russia; and Center of Integrable Systems, P.~G.~Demidov Yaroslavl State University, Sovetskaya 14, 150003, Yaroslavl, Russia}
\email{bbychkov@hse.ru}

\author[P.~Dunin-Barkowski]{Petr~Dunin-Barkowski}
\address{P.~D.-B.: Faculty of Mathematics, National Research University Higher School of Economics, Usacheva 6, 119048 Moscow, Russia; HSE--Skoltech International Laboratory of Representation Theory and Mathematical Physics, Skoltech, Nobelya 1, 143026, Moscow, Russia; and ITEP, 117218 Moscow, Russia}
\email{ptdunin@hse.ru}

\author[M.~Kazarian]{Maxim~Kazarian}
\address{M.~K.: Faculty of Mathematics, National Research University Higher School of Economics, Usacheva 6, 119048 Moscow, Russia; and Center for Advanced Studies, Skoltech, Nobelya 1, 143026, Moscow, Russia}
\email{kazarian@mccme.ru}

\author[S.~Shadrin]{Sergey~Shadrin}
\address{S.~S.: Korteweg-de Vries Institute for Mathematics, University of Amsterdam, Postbus 94248, 1090 GE Amsterdam, The Netherlands}
\email{S.Shadrin@uva.nl}	

\begin{abstract}
	We study a duality for the $n$-point functions in VEV formalism that we call the ordinary vs fully simple duality. It provides an ultimate generalisation and a proper context for the duality between maps and fully simple maps observed by Borot and Garcia-Failde. Our approach allows to transfer the algebraicity properties between the systems of $n$-point functions related by this duality, and gives direct tools for the analysis of singularities. As an application, we give a proof of a recent conjecture of Borot and Garcia-Failde on topological recursion for fully simple maps.
\end{abstract}

\maketitle	

\tableofcontents	
	
\section{Introduction}
An original motivation that lead to the results of this paper came from an attempt to understand the results and a conjecture of Borot and Garcia-Failde on maps and fully simple maps~\cite{borot2018simple} in the framework of algebraic formalism developed for the $n$-point functions of hypergeometric KP tau functions in~\cite{bychkov2021explicit,bychkov2020topological}.

Borot and Garcia-Failde studied the duality between maps and fully simple maps as an instance of a phenomenon known in the theory of topological recursion, the $x\leftrightarrow y$ duality, developed by Eynard and Orantin~\cite{eynard2013xy}. Borot and Garica-Failde also proved (see also a different fully combinatorial proof in~\cite{borot2019relating}) that enumerations of map and fully simple maps are related by a matrix whose coefficients are given by monotone / strictly monotone Hurwitz numbers.

The second result has a nice interpretation in terms of the vacuum expectation value (VEV) presentation of KP tau functions and using a VEV formula for ordinary maps it gives a VEV formula for fully simple maps. The relation between the VEV formulas for the tau functions of ordinary maps and fully simple maps gives rise to a very general duality for the $n$-point functions of KP tau functions and beyond, for more general VEV $n$-point functions that cover, for instance, the stuffed maps and their generalisations. Note that in general this generalised duality (which we still call the \emph{ordinary} vs \emph{fully simple} duality) is not reduced to the $x\leftrightarrow y$ duality.

This general duality 
is the main object of study in this paper. We prove explicit closed algebraic formulas (of the same flavour as in~\cite{bychkov2021explicit}) expressing the general \emph{ordinary} $n$-point functions in terms of the \emph{fully simple} ones 
and vice versa.


These algebraic relations between the $n$-point functions appear to be a very efficient tool to transfer the properties between the sides of this duality. In particular, we apply this technique to the case of fully simple maps, and we prove (using the known topological recursion on the side of ordinary maps) the topological recursion statement for fully simple maps, conjectured by Borot and Garcia-Failde in~\cite{borot2018simple}.

\subsection{Organisation of the paper} This paper heavily depends on the techniques developed in~\cite{bychkov2021explicit,bychkov2020topological}. For many explicit computations that literally repeat what is done in~\emph{op.~cit.} we rather give references than repeat them here.

Section~\ref{sec:VEVs} is devoted to explaining the vacuum expectation values that appear in other sections, and how to represent them as sums over graphs.

In Section~\ref{sec:mapsdef} we recall the details regarding the ordinary and the fully simple maps, and state the theorem on topological recursion for the fully simple maps (originally conjectured by Borot and Garcia-Failde in~\cite{borot2018simple}).

Section~\ref{sec:dual} is devoted to our main result, namely the generalised ordinary vs fully simple duality, which holds in a much more general context than that of maps. We formulate and prove it first (Sections \ref{sec:KPDuality-beginning}-\ref{sec:duality}) in the KP context (where we also pose a conjecture on topological recursion for the respective $n$-point functions), and then (Sections \ref{sec:BeyondKP}-\ref{sec:general duality}) in an even more general context inspired by the so-called stuffed maps .

In Section~\ref{sec:TRforFS} we use the main result of Section~\ref{sec:dual} in order to prove the theorem on topological recursion for fully simple maps (Theorem~\ref{th:BGF}) stated in Section~\ref{sec:mapsdef}.

\begin{remark}
	Throughout the paper we use the notation $\Mf_{g,n}$, $W_{g,n}$ and $\Mf^\vee_{g,n}$, $W^\vee_{g,n}$ for the ordinary and fully simple \emph{$n$-point functions}, respectively
	. Let us note here that we abuse the notation: in Sections~\ref{sec:mapsdef} and~\ref{sec:TRforFS} $\Mf_{g,n}$, $W_{g,n}$, $\Mf^\vee_{g,n}$, $W^\vee_{g,n}$ stand for the $n$-point functions of maps/fully simple maps, while in Sections~\ref{sec:KPDuality-beginning}-\ref{sec:duality} they stand for the general $t$-deformed hypergeometric KP $n$-point functions and in Sections~\ref{sec:VEVs} and \ref{sec:BeyondKP}-\ref{sec:general duality} they stand for an even more general case inspired by the stuffed maps. See the respective sections for the precise definitions.
\end{remark}

\subsection{Acknowledgements}
When this paper was essentially finished and we mentioned the results at the 7th Workshop on Combinatorics of Moduli Spaces, Cluster Algebras, and Topological Recursion (organized by L.~Chekhov, S.~Lando \emph{et al.} online on May 31--June 4 2021), we got a message from E.~Garcia-Failde that she together with G.~Borot and S.~Charbonnier also found an alternative independent proof of their conjecture with Borot on the topological recursion for fully simple maps, which was to appear soon (at the time of the current update of our present paper, Borot--Charbonnier--Garcia-Failde's paper is already available, \cite{BCGF21}).

B.~B. and M.~K. were partially supported by the International Laboratory of Cluster Geometry at the HSE University. S.~S. was supported by the Netherlands Organization of Scientific Research.

\section{Vacuum expectation values}
\label{sec:VEVs}
The main technical tool used in this paper allowing us to work with various kinds of correlator functions is to express them as
certain vacuum expectation values and then to reduce computations to sums over graphs. In this section we explain a general set-up for this technique.
We refer to~\cite{Kac,MJD} and also~\cite{bychkov2021explicit} as the main sources on the formalism presented here.

Consider the bosonic Fock space $\cF\coloneqq\C[[p_1,p_2,\dots]]$. Let $\vac\coloneqq 1\in\cF$ and define $\covac\colon \cF\to \C$ as $\covac \colon f(p_1,p_2,\dots) \mapsto f|_{p_i=0, i\geq 1}$. For $m>0$ denote
$J_m\coloneqq m\partial_{p_m}$ and $J_{-m}\coloneqq p_{m}$ and set $J_0\coloneqq0$.
For a formal power series $\psi(y,\hbar^2)$ such that $\psi(0,\hbar^2)=0$ we define an operator $\cD_\psi$ that acts diagonally in the basis of Schur functions as
\begin{equation}	\label{eq:Dpsidef}
\cD_\psi s_\lambda \coloneqq \exp\left(\sum_{(i,j)\in\lambda} \psi(\hbar(j-i),\hbar^2)\right)\,s_\lambda
\end{equation}
(below, once the functions $\psi_i$ are fixed and obvious from the context, we often omit them in the notation for $\cD$).

Given a collection of quantities $s_{g;k_1,\dots,k_m}$ defined for $g\ge0$, $m\ge 1$, $k_i\ge1$ and symmetric in $k_1,\dots,k_m$, we are interested in the VEVs of the type
\begin{align}
	W^\bullet_n(X_1,\dots,X_n) & \coloneqq \big\langle 0 \big|  \prod_{i=1}^n \cD_{\psi_i}^{-1} \Big(\sum_{\ell_i=1}^\infty J_{\ell_i} X_i^{\ell_i}\Big)  \cD_{\psi_i}
	\\ \notag
	& \quad
	\exp \Big({\sum_{g=0}^\infty \sum_{m=1}^{\infty} \frac{\hbar^{2g-2+m}}{m!} \sum_{k_1,\dots,k_m=1}^\infty s_{g;k_1,\dots,k_m} \prod_{j=1}^m \frac{J_{-k_j}}{k_j} }\Big)
	\big| 0 \big\rangle.
\end{align}
Here $\psi_i(y,\hbar^2)$, $i=1,\dots,n$, are formal power series such that $\psi_i(0,\hbar^2)=0$ that might eventually coincide with each other. Denote
\begin{align}
	\cA_i &  \coloneqq \cD_{\psi_i}^{-1} \Big(\sum_{\ell_i=1}^\infty J_{\ell_i} X_i^{\ell_i}\Big)  \cD_{\psi_i} ,
	& \cA_I & \coloneqq \cA_{i_1}\cdots\cA_{i_{|I|}}, &  I=\{i_1<i_2<\cdots<i_{|I|}\}\subseteq \llbracket n \rrbracket,
\end{align}
and
\begin{align}
	\Eop & \coloneqq\exp \Big({\sum_{g=0}^\infty \sum_{m=1}^{\infty} \frac{\hbar^{2g-2+m}}{m!} \sum_{k_1,\dots,k_m=1}^\infty s_{g;k_1,\dots,k_m} \prod_{j=1}^m \frac{J_{-k_j}}{k_j} }\Big). \end{align}
With this notation $W^\bullet_n(X_1,\dots,X_n) =  \VEV{\mathcal{A}_{\llbracket n \rrbracket} \Eop}$. We define a connected VEV as
\begin{align}\label{eq:inclexclinv}
	W_n(X_1,\dots,X_n) & =  \VEV{\mathcal{A}_{\llbracket n \rrbracket} \Eop}^\circ 
	\coloneqq \sum\limits_{l=1}^n\frac{(-1)^{l-1}}{l}\sum_{\substack{I_1\sqcup\ldots\sqcup I_l=\llbracket n \rrbracket \\ \forall j\, I_j \neq \emptyset}} \prod_{i=1}^l  \VEV{\mathcal{A}_{I_i}\Eop}.
\end{align}

The computation of such expressions is based on the following formula (\cite{bychkov2021explicit}, Section 3.1):
\begin{align} \label{eq:VertexOp}
	\cA_i = \sum_{m=1}^\infty X_i^m[z^m] \sum_{r=0}^\infty\partial_y^r e^{m \frac{\cS(m\hbar\partial_y)}{\cS(\hbar \partial_y)} \psi_i(y)}
	\big|_{y=0}[u^r]
	\frac{
		e^{\sum_{i=1}^\infty u\hbar\cS(u\hbar i) J_{-i}z^{-i}}
		e^{\sum_{i=1}^\infty u\hbar \cS(u\hbar i) J_{i}z^{i}}}
	{u\hbar\cS(u\hbar)},
\end{align}
as well as the obvious identities $J_{>0}\vac = 0$, $\covac J_{<0} = 0$, commutation relations $[J_a,J_b] = a\delta_{a+b,0}$, and their natural corollaries like
\begin{align}
	\Big[\sum_{i=1}^\infty J_i z^i, \sum_{j=1}^\infty J_{-j} w^{-j}\Big] = \frac{zw}{(z-w)^2},
\end{align}
where the latter formula is understood as its asymptotic expansion in the sector $|z|<|w|$.
In~\eqref{eq:VertexOp} and everywhere below in the paper for a power series $f(x)$ the expression $[x^k]f(x)$ stands for the coefficient in front of $x^k$ in the series, and also we have
\begin{equation}
	\cS(u):=\dfrac{e^{u/2}-e^{-u/2}}{u}.
\end{equation}

With these formulas, we expand $W_n^\bullet$ as an infinite sum of the coefficients of the VEVs of the type
\begin{align} \label{eq:mainingridient}
	\VEV{\prod_{i=1}^n 	\Big(e^{\sum_{l=1}^\infty u_i\hbar\cS(u_i\hbar l) J_{-l}z_i^{-l}}
		e^{\sum_{r=1}^\infty u_i\hbar \cS(u_i\hbar r) J_{r}z_i^{r} }\Big) \Eop},
\end{align}
We commute all operators with the negative indices to the left, and all operators with the positive indices to the right. This way we obtain that~\eqref{eq:mainingridient} is given as a sum over bi-coloured graphs (with white and black vertices) with the following properties:
\begin{itemize}
	\item There are $n$ vertices labelled from $1$ to $n$ (white vertices).
	\item There is a finite number of ordinary edges connecting the white vertices. Self-adjusted ordinary edges (loops) are forbidden. An ordinary edge connecting the white vertices labelled by $i$ and $j$, $i\not=j$, is decorated by
	\begin{align}
		u_i\hbar \cS(u_i\hbar z_i\partial_{z_i})u_j\hbar \cS(u_j\hbar z_j\partial_{z_j})\frac{z_iz_j}{(z_i-z_j)^2}.
	\end{align}
	\item There is a finite number of multi-edges (unlabelled black vertices connected by a finite number of edges to white vertices). A multi-edge can be connected to a white vertex multiple times. A multi-edge that is connected to the white vertices labelled by $i_1,\dots,i_k$ (the indices might be repeated) is decorated by
	\begin{align}
 \sum_{g=0}^\infty \hbar^{2g-2+k} \sum_{t_1,\dots,t_k=1}^\infty s_{g;t_1,\dots,t_k} \prod_{l=1}^k u_{i_l}\hbar \cS(u_{ i_l}\hbar z_{ i_l}\partial_{z_{i_l}}) z_{i_l}^{t_l}.
	\end{align}
	\item The value that we associate to a graph is the product of the weights on its ordinary and multi-edges divided by the order of its automorphism group.
\end{itemize}
Note that there is no requirement that the graphs are connected.

Denote the set of not necessarily connected graphs with $n$ labelled white vertices and with any number of ordinary and multi-edges by $\tilde \Gamma_n^\bullet$. Let $E(\gamma)$ be the set of ordinary and multi-edges of a graph $\gamma\in \tilde \Gamma_n^\bullet$, and denote the weight of an ordinary or a multi-edge $e$ by $\widetilde{\genwgt}(e)$. Then the formula for $W_n^\bullet$ is
\begin{align} \label{eq:disconnectedW}
	W_n^\bullet & =
	\sum_{\substack{r_1,\dots,r_n\geq 0\\ m_1,\dots,m_n\geq 1}}
		\prod_{i=1}^n X_i^{m_i} \partial_{y_i}^{r_i} e^{m_i\frac{\cS(m_i\hbar \partial_{y_i})}{\cS(\hbar\partial_{y_i})}\psi_i(y_i,\hbar^2)}\big|_{y_i=0} \prod_{i=1}^n [u_i^{r_i} z_i^{m_i}]
		\\ \notag & \quad
		 \prod_{i=1}^n\frac{1}{u_i\hbar S(u_i\hbar)} \sum_{\gamma\in \tilde\Gamma_n^\bullet} \frac {\prod_{e\in E(\gamma)} \widetilde{\genwgt}(e)}{|\mathrm{Aut}(\gamma)|}.
\end{align}

A direct corollary of the inclusion-exclusion formula~\eqref{eq:inclexclinv} is the following lemma:
\begin{lemma} Let $\tilde\Gamma_n\subset \tilde\Gamma_n^\bullet$ be a subset of the connected graphs. Then
	\begin{align} \label{eq:connectedW}
		W_n = \sum_{\substack{r_1,\dots,r_n\geq 0\\ m_1,\dots,m_n\geq 1}}
		\prod_{i=1}^n X_i^{m_i} \partial_{y_i}^{r_i} e^{m_i\frac{\cS(m_i\hbar \partial_{y_i})}{\cS(\hbar\partial_{y_i})}\psi_i(y_i,\hbar^2)}\big|_{y_i=0}  \prod_{i=1}^n [u_i^{r_i} z_i^{m_i}] \prod_{i=1}^n\frac{1}{u_i\hbar S(u_i\hbar)} \sum_{\gamma\in \tilde\Gamma_n} \frac {\prod_{e\in E(\gamma)} \widetilde{\genwgt}(e)}{|\mathrm{Aut}(\gamma)|}.
	\end{align}
\end{lemma}

In order to get more compact formulas for practical computations, it is also useful to reduce the sets of graphs $\tilde \Gamma_n^\bullet$ and $\tilde \Gamma_n$ to forbid the multiple edges and thus include coefficients coming from the automorphism group into the labels on the multi-edges. To this end, we introduce a subset $\Gamma_n^\bullet\subset \tilde \Gamma_n^\bullet$ of graphs that has at most one multi-edge for each tuple of indices $a_1,\dots,a_k$, where $a_1<\cdots<a_k$, for each $k=1,\dots,n$, and no ordinary edges (there are multi-edges with $2$ indices, and the contributions of former ordinary edges are incorporated in those). For an edge $e$ like that define $\genwgt(e)$ as
\begin{align}
	& \genwgt(e)\coloneqq -1+\exp\Bigg(\delta_{k,2}		u_{a_1}\hbar \cS(u_{a_1}\hbar z_{a_1}\partial_{z_{a_1}})u_{a_2}\hbar \cS(u_{a_2}\hbar z_{a_1}\partial_{z_{a_1}})\frac{z_{a_1}z_{a_2}}{(z_{a_1}-z_{a_2})^2} + \\ \notag
	&  \sum_{g=0}^\infty  \sum_{r_1,\dots,r_k=1}^\infty \frac{\hbar^{2g-2+r_1+\cdots+r_k}}{r_1!\cdots r_k!} \sum_{\substack{t_{ij}=1\\ i=1,\dots,k\\
	j=1,\dots,r_i}}^\infty s_{g;\{t_{ij}\}} \prod_{i=1}^k\Bigg(\prod_{j=1}^{r_j} u_{a_i}\hbar S\big(u_{a_i}\hbar z_{a_i}\partial_{ z_{a_i}}\big) z_{a_i}^{t_{ij}}\Bigg)\Bigg).
\end{align}

\begin{lemma} Let $\Gamma_n\subset \Gamma_n^\bullet$ denote the subset of connected graphs. We have:
	\begin{align}		 \label{eq:disconnectedW-nomultedges}
				W_n^\bullet & = \sum_{\substack{r_1,\dots,r_n\geq 0\\ m_1,\dots,m_n\geq 1}}
				\prod_{i=1}^n X_i^{m_i} \partial_{y_i}^{r_i} e^{m_i\frac{\cS(m_i\hbar \partial_{y_i})}{\cS(\hbar\partial_{y_i})}\psi_i(y_i,\hbar^2)}\big|_{y_i=0}
				 \prod_{i=1}^n [u_i^{r_i} z_i^{m_i}] \frac{\sum_{\gamma\in \Gamma_n^\bullet} \prod_{e\in E(\gamma)} {\genwgt}(e)}{\prod_{i=1}^n u_i\hbar S(u_i\hbar)} ;\\
		 \label{eq:connectedW-nomultedges}
		W_n &= \sum_{\substack{r_1,\dots,r_n\geq 0\\ m_1,\dots,m_n\geq 1}}
				\prod_{i=1}^n X_i^{m_i} \partial_{y_i}^{r_i} e^{m_i\frac{\cS(m_i\hbar \partial_{y_i})}{\cS(\hbar\partial_{y_i})}\psi_i(y_i,\hbar^2)}\big|_{y_i=0}
				 \prod_{i=1}^n [u_i^{r_i} z_i^{m_i}] \frac{\sum_{\gamma\in \Gamma_n} \prod_{e\in E(\gamma)} {\genwgt}(e)}{\prod_{i=1}^n u_i\hbar S(u_i\hbar)} .
	\end{align}
\end{lemma}
\section{Fully simple maps and topological recursion}\label{sec:mapsdef}

In this section, we recall the definitions of ordinary maps, fully simple maps, and their VEV expressions. We also recall the definition of topological recursion of Chekhov--Eynard--Orantin, and a statement on topological recursion of ordinary maps as well as a conjecture of Borot--Garcia-Failde on topological recursion for fully simple maps (the latter one is proved below in this paper).

\subsection{Geometric definition}
Let $\Map_{g;\mu_1,\dots,\mu_n}$ count the number of connected ordinary maps. More precisely, $\Map_{g;\mu_1,\dots,\mu_n}$ is a generating formal series in $t_1,t_2,t_3,\dots$ whose coefficient of $\prod_{i=1}^p t_i^{m_i}/m_i!$,  with $\sum_{i=1}^p m_i = M$, is equal to the weighted count of the number of ways to combinatorially glue a surface of genus $g$ from $n+M$ ordered polygons along their sides. The numbers of sides of the first $n$ polygons should be $\mu_1,\dots,\mu_n$, respectively (these are the distinguished polygons), and then there are $m_i$ polygons with $i$ sides, $i=1,\dots,p$. The weighted count means that we weight each polygonal decomposition with the inverse order of its symmetry group.

Let $\FSMap_{g;\mu_1,\dots,\mu_n}$ be the number of connected fully simple maps, which has exactly the same definition as ordinary maps, with one important addendum: consider the genus $g$ surface combinatorially glued from polygons as a cellular complex. We require that each $0$-cell is adjacent to at most one $1$-cell that is a boundary component of a distinguished polygon.

By $\Map^\bullet_{g;\mu_1,\dots,\mu_n}$ and $\FSMap^\bullet_{g;\mu_1,\dots,\mu_n}$ we denote the disconnected variants of the above definitions, the index $g$ here refers to a possibly disconnected surface of the Euler characteristic $2-2g$.

Note a slight difference between our definitions and the definitions in~\cite{borot2018simple,borot2019relating} --- in \emph{op.~cit.} the authors require also that each distinguished polygon also has one distinguished side, which amounts in a difference by a simple combinatorial factor. In the present paper we drop this requirement, to make formulas and computations below a bit nicer.

\subsection{VEV formulas} The interpretation of the enumeration of maps as coefficients of a particular hypergeometric KP tau function is proved in~\cite{goulden2008kp}. We have:
\begin{align}
	\Map^\bullet_{g;\mu_1,\dots,\mu_n} & = [\hbar^{2g-2+n}]\Vev{\prod_{i=1}^n \frac{J_{\mu_i}}{\mu_i} e^{\sum_{i=1}^\infty \frac{J_i t_i}{i\hbar}} \cD e^{\frac{J_{-2}}{2\hbar}}
	},
\end{align}
where $\cD\coloneqq\cD_\psi$ of \eqref{eq:Dpsidef} for
\begin{equation}
	\psi(y)=\log(1+y),
\end{equation}
i.e. $\cD v_\lambda = \prod_{(i,j)\in\lambda} (1+\hbar(j-i))$. We have a similar VEV formula for the fully simple maps.
\begin{lemma}\label{lem:FSviaHurwitz} We have
\begin{align}
	\FSMap^\bullet_{g;\mu_1,\dots,\mu_n} & = [\hbar^{2g-2+n}]\Vev{\prod_{i=1}^n \frac{J_{\mu_i}}{\mu_i} \cD^{-1} e^{\hbar^{-1}\sum_{i=1}^\infty J_i t_i} \cD e^{\frac{J_{-2}}{2\hbar}}
	}.
\end{align}	
\end{lemma}
\begin{proof}
The results of~\cite{borot2018simple,borot2019relating} express $\FSMap^\bullet_{\mu_1,\dots,\mu_n}=\sum_{g\in \mathbb{Z}} \hbar^{2g-2+n} \FSMap^\bullet_{g;\mu_1,\dots,\mu_n} $ as
\begin{align}
	\FSMap^\bullet_{\mu_1,\dots,\mu_n} = \sum_{k=1}^\infty \frac 1{k!} \sum_{\lambda_1,\dots,\lambda_k=1}^\infty {\lambda_1\cdots\lambda_k} H^{\leq}_{\mu_1,\dots,\mu_n;\lambda_1,\dots,\lambda_k}\Big|_{\hbar \to -\hbar}
	\Map^\bullet_{\lambda_1,\dots,\lambda_k},
\end{align}
where
\begin{align}
		\Map^\bullet_{\mu_1,\dots,\mu_n} & = \Vev{\prod_{i=1}^n \frac{J_{\mu_i}}{\mu_i} e^{\sum_{i=1}^\infty \frac{J_i t_i}{i\hbar}} \cD e^{\frac{J_{-2}}{2\hbar}}
	},
\end{align}
and $H^{\leq}_{\mu_1,\dots,\mu_n;\lambda_1,\dots,\lambda_k}$ are the monotone double Hurwitz numbers given by~\cite{guaypaquet20152d}
\begin{align}
	H^{\leq}_{\mu_1,\dots,\mu_n;\lambda_1,\dots,\lambda_k} = \Vev{
		\prod_{i=1}^n \frac{J_{\mu_i}}{\mu_i}\cD^{-1}
		\prod_{j=1}^k \frac{J_{\lambda_j}}{\lambda_j}
	}\Big|_{\hbar \to -\hbar}
\end{align}
(note some differences in normalization in our paper and~\cite{borot2019relating}).
Using this and the fact that on the Fock space
\begin{align}
\Id =  \sum_{k=0}^\infty \frac 1{k!} \sum_{\lambda_1,\dots,\lambda_k=1}^\infty
\prod_{j=1}^k {J_{-\lambda_j}}
\,\big| 0\! \big\rangle	\big\langle\! 0 \big|\,
\prod_{j=1}^k \frac{J_{\lambda_j}}{\lambda_j}
\end{align}
we obtain the statement of the lemma.
\end{proof}

We introduce the following generating functions:
\begin{align}\label{eq:Mveedef1}
	\Mf_{g,n}(X_1,\dots,X_n)&\coloneqq\sum_{\mu_1,\dots,\mu_n=1}^\infty \Map_{g;\mu_1,\dots,\mu_n}X_1^{\mu_1}\cdots  X_n^{\mu_n}\\ \nonumber
	&=[\hbar^{2g-2+n}]\sum_{m_1,\dots,m_n=1}^\infty\prod_{i=1}^n \dfrac{X_i^{m_i}}{m_i}
	\Vev{J_{m_1}\dots J_{m_n}e^{\sum_{i=1}^\infty \frac{J_i t_i}{i\hbar}} \cD
		e^{\frac{J_{-2}}{2\hbar}
	}}^\circ,\\ \label{eq:Mveedef}
	\Mf_{g,n}^\vee(w_1,\dots,w_n)&\coloneqq\sum_{\mu_1,\dots,\mu_n=1}^\infty \FSMap_{g;\mu_1,\dots,\mu_n}w_1^{\mu_1}\cdots  w_n^{\mu_n}\\ \nonumber
	&=[\hbar^{2g-2+n}]\sum_{m_1,\dots,m_n=1}^\infty\prod_{i=1}^n \dfrac{w_i^{m_i}}{m_i}
	\Vev{J_{m_1}\dots J_{m_n}\cD^{-1}e^{\sum_{i=1}^\infty \frac{J_i t_i}{i\hbar}} \cD
		e^{\frac{J_{-2}}{2\hbar}
	}}^\circ
\end{align}
and
\begin{align}
	W_{g,n}(X_1,\dots,X_n)&\coloneqq D_{X_1}\cdots D_{X_n} \Mf_{g,n}(X_1,\dots,X_n), \label{def:Wgn1}\\
	W_{g,n}^\vee(w_1,\dots,w_n)&\coloneqq D_{w_1}\cdots D_{w_n} \Mf_{g,n}^\vee(w_1,\dots,w_n),\label{def:Wgn2}
\end{align}
where $D_{X_i}\coloneqq X_i\, \partial/\partial X_i$, $D_{w_i}\coloneqq w_i\, \partial/\partial w_i$.



\subsection{Formulation of topological recursion}

In this paper, it is more convenient for us to give not a general definition of topological recursion~\cite{EynardOra-TopoRec}, but its reformulation in terms of the loop equations and the projection property~\cite{BorotSha-Blobbed}, further specialized to the case of the underlying curve of genus $0$ and rational functions on it as in~\cite{bychkov2020topological}.

Let $X$ and $y$ be meromorphic functions on $\Sigma=\mathbb{C}P^1$ with affine coordinate $z$, and assume that the finite critical points of $X$ are $p_1,\dots,p_N$, all these critical points are simple; $y$ is regular at these points and has non-zero differential. Let $\sigma_i$ denote the deck transformation of $X$ near $p_i$. Let $B(z_1,z_2)\coloneqq dz_1dz_2/(z_1-z_2)^2$.

Consider a system of symmetric rational functions $\Mf_{g,n}(z_1,\dots,z_n)$, $g\geq 0$, $n\geq 1$, and set
\begin{align}
W_{g,n}&\coloneqq D_{1}\dots D_{n} \Mf_{g,n},\quad D_{i}\coloneqq X(z_i)\tfrac{d}{dX(z_i)},\\
\omega_{g,n}
&\coloneqq d_1\dots d_n \Mf_{g,n}+\delta_{g,0}\delta_{n,2}\tfrac{dX(z_1) dX(z_2)}{(X(z_1)-X(z_2))^2}\\
&=W_{g,n}\prod_{i=1}^n d\,\log(X(z_i))+\delta_{g,0}\delta_{n,2}\tfrac{dX(z_1) dX(z_2)}{(X(z_1)-X(z_2))^2}.\notag
\end{align}

\begin{definition} We say that the rational symmetric differentials $\omega_{g,n}$, $g\geq 0$, $n\geq 1$, satisfy the topological recursion on the spectral curve $(\Sigma,X(z),y(z),B(z_1,z_2))$ if
\begin{itemize}
	\item \emph{(Initial conditions)} We have:
	\begin{align}
		\omega_{0,1}(z) & = y(x) d\log X(z); & \omega_{0,2}(z_1,z_2) &= B(z_1,z_2).
	\end{align}
These equalities are also equivalent to
\begin{align}
W_{0,1}(z)&=y(z),&\Mf_{0,2}(z_1,z_2)&=\log\frac{z_1^{-1}-z_2^{-1}}{X(z_1)^{-1}-X(z_2)^{-1}}.
\end{align}
	\item \emph{(Linear Loop Equations)} For any $g,n\geq0$ the function $W_{g,n+1}$ may have a pole at $p_i$ for $i=1,\dots,N$ with respect to the first argument such that its principal part is skew symmetric with respect to the involution $\sigma_i$. In other words, we have that
	\begin{align}
		W_{g,n+1}(z,z_{\llbracket n\rrbracket})+ W_{g,n+1}(\sigma_i(z),z_{\llbracket n\rrbracket})		
	\end{align}
	is holomorphic in $z$ at $z\to p_i$.
	\item \emph{(Quadratic Loop Equations)} For any $g,n\geq 0$
the rational function
	\begin{align}
		W_{g-1,n+2}(z,z,z_{\llbracket n\rrbracket})+\sum_{\substack{g_1+g_1=g\\ I_1\sqcup I_2 = \llbracket n \rrbracket}}
W_{g_1,|I_1|+1}(z,z_{I_1})	W_{g_2,|I_2|+1}(z,z_{I_2})	
	\end{align}
may have a pole at $p_i$ in $z$ such that its principal part is skew symmetric with respect to the involution $\sigma_i$ for any $i=1,\dots,N$.
	\item \emph{(Projection Property)} For any $g,n\geq 0$ the function $\Mf_{g,n+1}(z,z_{\llbracket n\rrbracket})$ considered as a rational function in $z$ has no poles other than $p_1,\dots,p_N$. Equivalently, the differential $\omega_{g,n+1}(z,z_{\llbracket n\rrbracket})$ has no poles in $z$ other than $p_1,\dots,p_N$.
\end{itemize}
\end{definition}

The projection property implies that $W_{g,n+1}(z,z_{\llbracket n\rrbracket})$ also has no poles in~$z$ other than $p_1,\dots,p_N$. But this weaker condition on $W_{g,n+1}$ is not sufficient for the projection property. We need a stronger requirement that $W_{g,n+1}$ vanishes at the poles of $dX/X$ so that $W_{g,n+1}(z,z_{\llbracket n\rrbracket}) \tfrac{dX(z)}{X(z)}$ is holomorphic outside $p_1,\dots,p_N$.

The linear and quadratic loop equations together determine the principal part of the poles of $W_{g,n+1}(z,z_{\llbracket n\rrbracket})$ in $z$ up to the terms with at most simple poles at $p_i$. Multiplying by $dX(z)/X(z)$ we obtain the principal part of the poles of $\omega_{g,n+1}(z,z_{\llbracket n\rrbracket})$ without ambiguity. Any meromorphic differential form on $\mathbb{C}P^1$ is determined uniquely by the principal parts of its poles. The principal relation of topological recursion is nothing but an explicit formula expressing the form $\omega_{g,n+1}$ in terms of the principal parts of its poles.

\begin{remark}\label{rem:origrec}
	The original formulation~(\cite{EynardOra-TopoRec}) of the spectral curve topological recursion, which is equivalent to the one given above, allows one to explicitly recursively reconstruct all the $n$-point differentials $\omega_{g,n}$ starting with the spectral curve data $(\Sigma,X(z),y(z),B(z_1,z_2))$; we are not listing the respective formulas here for brevity.
\end{remark}

\subsection{Topological recursion for ordinary maps} Denote $\bt\coloneqq(t_1,\dots,t_q,0,0,...)$.
Set
\begin{align}
	x(z)&\coloneqq\alpha+\gamma\;\left(z+\frac1z\right),\\
	V(z)&\coloneqq x-\sum_{k=1}^qt_kx^{k-1}\bigm|_{x=x(z)},\\ \label{eq:wmapsdef}
	w(z)&\coloneqq V(z)_+\qquad\text{(\,the polynomial part of the Laurent polynomial } V(z)\text{ )},\\
	X(z)&\coloneqq\frac{1}{x(z)}=\frac{z}{\alpha\,z+\gamma\,(1+z^2)},\label{eq:wXy2}\\
	y(z)&\coloneqq\frac{w(z)}{X(z)}-1\label{eq:wXy3},
\end{align}
where $\alpha=\alpha(\bt)$ and $\gamma=\gamma(\bt)$ are functions defined in the vicinity of the point $\bt=0$ via implicit algebraic equations
\begin{equation}
	[z^0]V(z)=0,\qquad [z^{-1}] V(z)=\gamma^{-1},
\end{equation}
such that $\alpha(0)=0$, $\gamma(0)=1$.

With this set-up, the $n$-point functions $W_{g,n}(X_1,\dots,X_n)$ enumerating maps and defined by~\eqref{def:Wgn1} become rational functions in $z_1,\dots,z_n$ after the substitution $X_i=X(z_i)$. Moreover, the differentials
\begin{equation}
	\omega_{g,n}=W_{g,n}\prod_{i=1}^n\frac{dX_i}{X_i}+\delta_{g,0}\delta_{n,2}\frac{dX_1dX_2}{(X_1-X_2)^2}
\end{equation}
can be obtained via topological recursion (cf. Remark~\ref{rem:origrec}) with the initial data
\begin{equation}
	\omega_{0,1}(z)=y(z)\frac{dX(z)}{X(z)},\qquad \omega_{0,2}=\frac{dz_1dz_2}{(z_1-z_2)^2},
\end{equation}
i.e. on the spectral curve
\begin{equation}\label{eq:mapssec}	
(\mathbb{C}P^1,X(z),y(z),dz_1dz_2/(z_1-z_2)^2),
\end{equation} 
see~\cite{EynardTopoExpMM,ChekhovEynard,CEO,DOPS,EynardBook}.

In particular, this implies that for stable $(g,n)$ all possible poles of rational functions $W_{g,n}$ lie on the hyperplanes $z_i=\pm1$. When we say that $W_{g,n}$ are rational, we mean that they are ratios of polynomials of fixed degree whose coefficients are formal power series in the $t_i$ parameters. Moreover, the coefficients can be expressed as polynomials in $t_1,\dots,t_q,\alpha,\gamma$, so that the entire ``nonpolynomiality'' is contained in the algebraic functions $\alpha(\bt)$ and $\gamma(\bt)$. In particular, all these series have non-zero radii of convergence.

\subsection{Topological recursion for fully simple maps}
Given that topological recursion is known for $W_{g,n}(X_1,\dots,X_n)$, we prove it for $W_{g,n}^\vee(w_1,\dots,w_n)$, with the spectral curve data given by the functions $w(z)$ 
 and $y(z)$ defined by~\eqref{eq:wmapsdef}--\eqref{eq:wXy3}.

\begin{theorem}[Conjectured by Borot and Garcia-Failde,~\cite{borot2018simple}]\label{th:BGF}
	Let the parameters $t_i$, $i=1\dots q$, be such that the function $w(z)$ defined in~\eqref{eq:wmapsdef} has only simple critical points. Then the $n$-point functions
	$W_{g,n}^\vee(w_1,\dots,w_n)$ defined in~\eqref{def:Wgn2} satisfy the topological recursion on the spectral curve $(\mathbb{C}P^1,w(z),y(z),dz_1dz_2/(z_1-z_2)^2)$, where $y(z)$ is defined in~\eqref{eq:wXy3}; that is, the $(0,1)$-differential is given by $y(z)\tfrac{dw(z)}{w(z)}$.
	
	In terms of of the spectral curve data $X(z)$ and $y(z)$ for ordinary maps,~\eqref{eq:mapssec}, the spectral curve data for fully simple maps takes the following form:
	\begin{equation}
		w=X(z)(1+y(z)),\qquad y=y(z).
	\end{equation}

\end{theorem}
\begin{proof}
	Follows from Corollaries~\ref{cor:lqloop} and \ref{cor:MTheta}.
\end{proof}
\begin{remark}	
	If $w(z)$ has non-simple critical points, the ordinary topological recursion is not applicable, and one needs to use the Bouchard-Eynard recursion~\cite{BE13}. For brevity, we omit this case here, but the fact that the Bouchard-Eynard recursion holds for it, can be proved by taking a limit from the general case (where the ordinary topological recursion applies), analogously to \cite[Section 2]{BDS}.
\end{remark}

\section{Dualities}\label{sec:dual}
In this section we formulate and prove the generalised ordinary vs fully simple duality. We start with formulating and proving it in the KP context (along with introducing a conjecture on the topological recursion in this context) in Sections \ref{sec:KPDuality-beginning}-\ref{sec:duality}, and then, in Sections \ref{sec:BeyondKP}-\ref{sec:general duality}, we also do it in a more general context inspired by the so-called stuffed maps.
\subsection{Objects} \label{sec:KPDuality-beginning}
Consider three sets of \emph{parameters} $(c_1,c_2,\dots)$, $(s_1,s_2,\dots)$, $(t_1,t_2,\dots)$. For simplicity, we restrict ourselves to the ACEH-type~\cite{ACEH} case when only a finite number of parameters are different from zeroes.
\begin{remark}
	Note that there were not any $t$-parameters in \cite{ACEH}, that paper covered only the $t=0$ specialization of what is described below; still, we use the terminology \emph{ACEH-type} to reflect that only a finite number of parameters are non-zero, since this was the case for the $c$- and $s$-parameters in \cite{ACEH}.
\end{remark}	
Along with the parameters, we also have \emph{variables} $(p_1,p_2,\dots)$. Let
\begin{align}
	\psi(y)&\coloneqq\log(1+ c_1 y+c_2y^2+\dots),\\ \label{eq:phidef}
	\phi(y)&\coloneqq\exp(\psi(y))=1+ c_1 y+c_2y^2+\dots,\\ \label{eq:phimdef}
	\phi_m(y)
		&\coloneqq\exp\left(\sum_{i=1}^{m}\psi\left(y+\dfrac{2i-m-1}{2} \hbar\right)\right)
		,&m>0,\\ \label{eq:phi0def} 
		\phi_0(y)&\coloneqq1,&\\ 
		\phi_m(y)&\coloneqq(\phi_{-m}(y))^{-1},&m<0
\end{align}
(here, for simplicity, $\psi$ does not depend on $\hbar$, compared to what we had in Section~\ref{sec:VEVs}). Then
\begin{align}
	\cD_\psi s_\lambda &\mathop{\coloneqq}^{\eqref{eq:Dpsidef}} \exp\left(\sum_{(i,j)\in\lambda} \psi(\hbar(j-i))\right)\,s_\lambda.
\end{align}

Consider the following two partition functions
\begin{align}\label{eq:Zcorr}
Z& =e^{F}\coloneqq\Vev{e^{\sum_{i=1}^\infty \frac{J_i p_i}{i\hbar}}e^{\sum_{i=1}^\infty \frac{J_i t_i}{i\hbar}} \cD_\psi e^{\sum_{i=1}^\infty \frac{s_i J_{-i}}{i\hbar}}},
\\\label{eq:Zveecorr}
Z^\vee& =e^{F^\vee}\coloneqq\Vev{e^{\sum_{i=1}^\infty \frac{J_i p_i}{i\hbar}}\cD_\psi^{-1}e^{\sum_{i=1}^\infty \frac{J_i t_i}{i\hbar}} \cD_\psi e^{\sum_{i=1}^\infty \frac{s_i J_{-i}}{i\hbar}}}.
\end{align}
The first one specializes for $t=0$ to the ACEH partition function studied in~\cite{ACEH}, $Z^{\rm ACEH}=Z|_{t=0}$. On the other hand, for $t\ne0$ it is recovered from the ACEH function by a shift of variables,
\begin{equation}
	Z(p,t)=Z^{\rm ACEH}(p_1+t_1,p_2+t_2,\dots)=Z|_{t_i=0,p_i=p_i+t_i}.
\end{equation}

The $n$-point functions are then defined as follows: 
\begin{align}
W_{g,n}(X_1,\dots,X_n)&\coloneqq[\hbar^{2g-2+n}]\sum_{m_1,\dots,m_n=1}^\infty\frac{\partial^n F}{\partial p_{m_1}\dots\partial p_{m_n}}|_{p=0}\prod_{i=1}^n m_i X_i^{m_i}\label{eq:Wcorr}\\
&=[\hbar^{2g-2+n}]\sum_{m_1,\dots,m_n=1}^\infty\prod_{i=1}^n X_i^{m_i}
 \Vev{J_{m_1}\dots J_{m_n}e^{\sum_{i=1}^\infty \frac{J_i t_i}{i\hbar}} \cD_\psi
  e^{\sum_{i=1}^\infty \frac{s_i J_{-i}}{i\hbar}}}^{\circ}, \notag\\
W_{g,n}^\vee(w_1,\dots,w_n)&\coloneqq[\hbar^{2g-2+n}]\sum_{m_1,\dots,m_n=1}^\infty\frac{\partial^nF ^\vee}{\partial p_{m_1}\dots\partial p_{m_n}}|_{p=0}\prod_{i=1}^n m_i w_i^{m_i}\label{eq:Wveecorr}\\
&=[\hbar^{2g-2+n}]\sum_{m_1,\dots,m_n=1}^\infty  \prod_{i=1}^nw_i^{m_i}
 \Vev{J_{m_1}\dots J_{m_n}\cD_\psi^{-1}e^{\sum_{i=1}^\infty \frac{J_i t_i}{i\hbar}} \cD_\psi
  e^{\sum_{i=1}^\infty \frac{s_i J_{-i}}{i\hbar}}}^{\circ}. \notag
\end{align}

\begin{notation}	
We call the partition functions and the $n$-point functions defined above in this subsection (and used up until Section~\ref{sec:BeyondKP}) \emph{KP-type} as opposed to the more general type discussed in Section~\ref{sec:BeyondKP}.
\end{notation}
\begin{remark}	
In the case $t=0$ the coefficients of $W_{g,n}$ enumerate various kinds of Hurwitz numbers for various 
specializations of the other parameters. In the most general case when $c$ and $s$ parameters are not specified, these are the so-called weighted double Hurwitz numbers, see e.g.~\cite{ACEH}. The insertion of $t$-parameters corresponds to additional markings of ramified coverings enumerated by Hurwitz numbers. Therefore, we can refer to the coefficients of $W_{g,n}$ as the $t$-deformed weighted Hurwitz numbers. On the other hand, the functions $W_{g,n}^\vee$ are introduced in a formal way by an explicit formula~\eqref{eq:Wveecorr} and combinatorial meaning of their coefficients is not clear at the moment. The only case when the combinatorial meaning is assigned is the case of enumeration of (hyper)maps and fully simple (hyper)maps. In this way, we consider the coefficients of $W_{g,n}^\vee$ as the ``fully simple analogues'' of ($t$-deformed) weighted double Hurwitz numbers.
\end{remark}

\subsection{Conjectural topological recursions}
We conjecture that the $n$-point functions $W_{g,n}$ and $W_{g,n}^\vee$ satisfy topological recursions on some spectral curve.
\begin{conjecture}\label{conj:toprec}
	Consider the $n$-point functions $W_{g,n}(X_1,\dots,X_n)$ and $W^\vee_{g,n}(w_1,\dots,w_n)$ defined in \eqref{eq:Wcorr} and \eqref{eq:Wveecorr}, respectively. Here it is important that only a finite number of parameters $c,\,t,\,s$ are non-zero, as we have assumed at the beginning of Section~\ref{sec:KPDuality-beginning}; moreover we assume that these parameters are in \emph{general position}, more on that below. Also consider the curve  $\Sigma= {\mathbb C} P^1$ 
	with the global coordinate $z$, and a bidifferential $B(z_1,z_2)\coloneqq dz_1dz_2/(z_1-z_2)^2$.
	
	Then there exist meromorphic functions $X=X(z)$, $w=w(z)$, and $y=y(z)$ on $\Sigma$ such that
	\begin{itemize}
		\item $X(z)$ and $w(z)$ serve as local coordinates near $z=0$, i.e.
		\begin{equation}
			X(0)=w(0)=0,\qquad dX/dz(0)\neq 0,\; dw/dz(0)\neq 0.
		\end{equation}
		\item The $n$-point differentials $\omega_{g,n}$ and $\omega^\vee_{g,n}$, reconstructed via topological recursion (as in Remark~\ref{rem:origrec}) for the spectral curve data $(\Sigma,X(z),y(z),B(z_1,z_2))$ and $(\Sigma,w(z),y(z),B(z_1,z_2))$ respectively, have the power expansions at $z=0$ given by
	\begin{align}
		\omega_{g,n}&=W_{g,n}(X_1,\dots,X_n)\prod_{i=1}^n\frac{dX_i}{X_i}+\delta_{g,0}\delta_{n,2}\frac{dX_1dX_2}{(X_1-X_2)^2},\\
		\omega_{g,n}^\vee&=W_{g,n}^\vee(w_1,\dots,w_n)\prod_{i=1}^n\frac{dw_i}{w_i}+\delta_{g,0}\delta_{n,2}\frac{dw_1dw_2}{(w_1-w_2)^2},
	\end{align}
	respectively. Here $\forall i\;\; X_i=X(z_i),\; w_i=w(z_i)$.
	\item Moreover, the functions $X(z)$, $w(z)$ and $y(z)$ are related by the identity
	\begin{equation}	
		w=X\,\phi(y),
	\end{equation}
	where $\phi$ is the one of \eqref{eq:phidef}.	
\end{itemize}
	
	
	
	The aforementioned general position requirement is needed since for some special values of parameters $c,\,t,\,s$ the functions $X(z)$, $w(z)$ will acquire non-simple critical points, and thus the ordinary topological recursion will not be applicable; one has to use the Bouchard-Eynard recursion~\cite{BE13} for these special situations, but we are not giving the respective precise statement here for brevity.
	
\end{conjecture}
\begin{remark}
	For the two topological recursions conjectured above we have:
	\begin{itemize}
		\item The `poles' of the two recursions and the involutions near the poles are provided by the functions $X$ and $w$, respectively, considered as ramified coverings $\Sigma\to {\mathbb C} P^1$.
		\item The $(0,2)$ differential is the standard bidifferential $\omega_{0,2}=\omega^\vee_{0,2}=B(z_1,z_2)=\frac{dz_1dz_2}{(z_1-z_2)^2}$.
		\item The $(0,1)$-differentials are given by
		\begin{align}
			\omega_{0,1}&=y\,\frac{dX}{X},\\
			\omega^\vee_{0,1}&=y\,\frac{dw}{w},
		\end{align}
		in other words, the functions $W_{0,1}=W_{0,1}^\vee=y$ considered as functions on the spectral curve coincide.
		\end{itemize}
\end{remark}

\begin{remark}	
The conjecture in this form remains open, even just for the case of the ``ordinary'' functions $W_{g,n}$ (disregarding the ``fully simple'' ones, $W^\vee_{g,n}$). The topological recursion statement is proved for $W_{g,n}$ for the specialization $t=0$ only, see~\cite{ACEH},~\cite{bychkov2020topological}. 
We expect that this statement extends to the $t$-deformed case with the same number of critical points for the function $X$. Meanwhile, we can prove certain duality statements for $W_{g,n}$ and $W_{g,n}^\vee$, see below (Theorems~\ref{th:WWvee} and~\ref{th:WveeW}), which can potentially allow one to prove the topological recursion on one side if it is already known on the other side independently. We then do the latter for the maps / fully simple maps case, i.~e., we prove the Borot--Garcia-Failde conjecture (Theorem~\ref{th:BGF}) this way.
\end{remark}

\subsection{Duality formulation}
In this subsection we formulate the duality statement for the $W_{g,n}$ and $W^\vee_{g,n}$. For this we need to introduce \emph{multigraphs}:
\begin{notation}\label{not:multigraphs}	
Let $V$ be a finite set. A \emph{multiedge} is an arbitrary subset of $V$. A \emph{multigraph} with the set of vertices $V$ is an arbitrary collection of multiedges. A multigraph can be represented as a bipartite graph with white and black vertices. The white vertices are considered as `true' vertices while black ones are regarded as multiedges. Every multiedge connects several white vertices, namely, those connected with the corresponding black vertex by edges in the usual sense. A multigraph is called \emph{connected} if the corresponding bipartite graph is connected. All multigraphs that we are considering are \emph{simple} meaning that neither multiple multiedges nor multiple edges in the corresponding bipartite graph are allowed. We denote by $\Gamma_n$ the set of all connected multigraphs on $n$ numbered (white) vertices. For a multigraph $\gamma\in\Gamma_n$ we denote by $V(\gamma)$ and $E(\gamma)$ the set of its vertices and multiedges, respectively. 
\end{notation}
\begin{remark}
	Notation~\ref{not:multigraphs} just gives a more detailed and formal definition of the same set $\Gamma_n$ which appeared at the end of Section~\ref{sec:VEVs}; the graphs from this set differ from the graphs from the set $\tilde\Gamma_n$ of that section by lacking any multiple edges or ``ordinary'' edges connecting white vertices. 
\end{remark}

We formulate an expression for $W_{g,n}$ as a sum over multigraphs on $n$ vertices involving $W^\vee_{g',n'}$'s for different pairs of ${g',n'}$ and a similar expression for $W^\vee_{g,n}$ in terms of $W_{g',n'}$'s. These expressions have a symmetric form providing the duality between the functions $W_{g,n}$ and $W^\vee_{g,n}$: the collection of one of these two functions determines uniquely the other one.

Define
\begin{align}\label{eq:first}
\tilde\delta_j^\vee&\coloneqq\sum_{k=1}^\infty u_j\hbar\,k\,\cS(u_j\,\hbar\, k)w_j^k\frac{\partial}{\partial p_k},\\
\lwgtpart^\vee_{j_1,\dots,j_k}&\coloneqq\tilde\delta^\vee_{j_1}\dots\tilde\delta^\vee_{j_k} F^\vee\bigm|_{p=0}\\
&=\Biggl(\prod_{i=1}^k u_{j_i}\hbar\,\cS(u_{j_i}\,\hbar w_{j_i}\partial_{w_{j_i}})\Biggr)
   \sum_{g=0}^\infty \hbar^{2g-2+k}W^\vee_{g,k}(w_{j_1},\dots,w_{j_k}),\notag\\
\lwgt^\vee_{j_1,\dots,j_k}&\coloneqq\left(\prod_{i=1}^k(e^{\tilde\delta^\vee_{j_i}}-1)\right)F^\vee\bigm|_{p=0}\\
&=\sum_{r_1,\dots,r_k=1}^\infty\tfrac{1}{r_1!\dots r_k!}\lwgtpart^\vee_{\underbrace{\scriptstyle j_1,\dots,j_1}_{r_1},\dots,\underbrace{\scriptstyle j_k,\dots,j_k}_{r_k}}\notag.
\end{align}

Let $\gamma\in \Gamma_n$ be a multigraph. Let $e=(j_1,\dots,j_k)\in E(\gamma)$ be a multiedge connecting vertices $j_1,\dots,j_k$. We define the weight of this multiedge as follows. If $k>2$, we set
\begin{align}
\wgt^\vee(j_1,\dots,j_k)\coloneqq&\exp\left(\left(\prod_{i=1}^k(e^{\tilde\delta^\vee_{j_i}}-1)\right)F^\vee\bigm|_{p=0}\right)-1\\
=&e^{\lwgt^\vee_{j_1,\dots,j_k}}-1.\notag
 \end{align}
In the exceptional cases $k\le2$ these formulas are slightly modified. Namely, if $k=2$, we set
\begin{align}
\lwgt^{\vee\mathrm{corr}}_{i,j}\coloneqq&\hbar^2u_{i}u_{j}\cS(u_{i}\hbar w_{i}\partial_{w_{i}})\cS(u_{j}\hbar w_{j}\partial_{w_{j}})\frac{w_{i} w_{j}}{(w_{i}-w_{j})^2},\\
\wgt^\vee(i,j)\coloneqq&e^{\lwgt^\vee_{i,j}+\lwgt^{\vee\mathrm{corr}}_{i,j}}-1.\
 \end{align}
Finally, in the case $k=1$ we set
\begin{align}
\wgt^\vee(j)\coloneqq&e^{\lwgt^\vee_{j}-u_j\,W^\vee_{0,1}(w_j)}.
 \end{align}
\begin{remark}
	This $\wgt^\vee(j_1,\dots,j_k)$ (and $\wgt(j_1,\dots,j_k)$ below) is similar to the $\genwgt$ we had in Section~\ref{sec:VEVs} (or, more precisely the former is a special case of the latter), except that it gets this correction for $k=1$, which was not the case for the $\genwgt$.
\end{remark}

With these notations, the main relation connecting the KP-type `ordinary' and 'fully simple' correlator functions reads as follows:
\begin{theorem}\label{th:WWvee}
\begin{align}\label{eq:Wgn}
W_{g,n}(X_1,\dots,X_n)&{}+\delta_{g,0}\delta_{n,2}\tfrac{X_1X_2}{(X_1-X_2)^2}=[\hbar^{2g-2+n}]U_1\dots U_n\sum_{\gamma\in\tilde\Gamma_n}
  \prod_{e\in E(\gamma)}\wgt^\vee(e),
\end{align}
where $U_i$ is the transformation sending a function $H(u_i,w_i)$ in $u_i$ and $w_i$ to the function $U_iH$ in $X_i$ given by
\begin{equation}\label{eq:Uorig}
(U_iH)(X_i)\coloneqq\sum_{r=0}^\infty \sum_{m=-\infty}^\infty
  X_i^m\;(\partial_y^{r}\phi_m(y)|_{y=0})\;[u_i^r w_i^m] e^{u_i W^\vee_{0,1}(w_i)}\frac{H(u_i,w_i)}{u_j\hbar\cS(u_j\hbar)}.
\end{equation}
If $H(u_i,w_i)/u_i$ is regular in $u_i$ (which is always the case for $n\geq 2$), we also have
\begin{align}\label{eq:Ured}
	(U_iH)(X_i)&=\sum_{j=0}^\infty (\tfrac{1}{Q^\vee(w_i)}w_i\partial_{w_i})^j\sum_{r=0}^\infty
	\tfrac{[v^j]L_r(v,\hbar,W^\vee_{0,1}(w_i))}{Q^\vee(w_i)}
	[u_i^r]\frac{H(u_i,w_i)}{u_j\hbar\cS(u_j\hbar)}\Bigm|_{w_i=w(X_i)},\\
	L_{r}(v,\hbar,y) &\coloneqq \left(\partial_y+v\psi'(y)\right)^r e^{v\left(\frac{\cS(v\hbar\partial_y)}{\cS(\hbar\partial_y)}-1\right)\psi(y)},\\
	Q^\vee(w_i)&\coloneqq\frac{w_i}{X_i}\frac{dX_i}{dw_i}=1-w_i\partial_{w_i}\log(\phi(W^\vee_{0,1}(w_i))).\label{eq:last}
\end{align}
\end{theorem}
The form \eqref{eq:Ured} can be obtained from \eqref{eq:Uorig} due to the fact that the `principal identity' of \cite{bychkov2021explicit} implies 
 that the action of $U_i$ can be expressed in this case through the local change of variables $w=w(X)$ defined by an implicit equation
\begin{equation}\label{eq:wX}
X=\frac{w}{\phi(W^\vee_{0,1}(w))}.
\end{equation}
The advantage of this form of $U_i$ is that when it is applied in~\eqref{eq:Wgn} the sums over $r$ and $j$ appear to be \emph{finite} for any fixed~$g$ and~$n$.

All relations \eqref{eq:first}--\eqref{eq:last} have dual analogues with the exchange of the generalised `ordinary' and `fully simple' functions. Namely, denoting
\begin{align}
\lwgtpart_{j_1,\dots,j_m}
&\coloneqq\Biggl(\prod_{i=1}^m u_{j_i}\hbar\,\cS(u_{j_i}\,\hbar X_{j_i}\partial_{X_{j_i}})\Biggr)
   \sum_{g=0}^\infty \hbar^{2g-2+m}W_{g,m}(X_{j_1},\dots,X_{j_m}),\\
\lwgt_{j_1,\dots,j_k}&\coloneqq\sum_{r_1,\dots,r_k=1}^\infty\tfrac{1}{r_1!\dots r_k!}\lwgtpart_{\underbrace{\scriptstyle j_1,\dots,j_1}_{r_1},\dots,\underbrace{\scriptstyle j_k,\dots,j_k}_{r_k}},\\
\lwgt^{\mathrm{corr}}_{i,j}&\coloneqq\hbar^2u_{i}u_{j}\cS(u_{i}\hbar X_{i}\partial_{X_{i}})\cS(u_{j}\hbar X_{j}\partial_{X_{j}})\frac{X_{i} X_{j}}{(X_{i}-X_{j})^2},\\
\wgt(j_1,\dots,j_k)&\coloneqq
e^{\lwgt_{j_1,\dots,j_k}+\delta_{k,2}\lwgt^{\mathrm{corr}}_{j_1,j_2}-\delta_{k,1} u_{j_1}W_{0,1}(X_{j_1})}-1+\delta_{k,1},
\quad (j_1,\dots,j_k)\in E(\gamma),
 \end{align}
we have:

\begin{theorem}\label{th:WveeW}
\begin{align}\label{eq:Wveegn}
W^\vee_{g,n}(w_1,\dots,w_n)&{}+\delta_{g,0}\delta_{n,2}\tfrac{w_1w_2}{(w_1-w_2)^2}=[\hbar^{2g-2+n}]U^\vee_1\dots U^\vee_n\sum_{\gamma\in\tilde\Gamma_n}
  \prod_{e\in E(\gamma)}\wgt(e).
\end{align}
The transformation $U^\vee_i$ takes the function $H(u_i,X_i)$ in $u_i$ and $X_i$ to the function $U^\vee_iH$ in $w_i$ given by
\begin{equation}
\label{eq:Uvee}
(U^\vee_iH)(w_i)\coloneqq\sum_{r=0}^\infty \sum_{m=-\infty}^\infty
  w_i^m\;(\partial_y^{r}\phi_{-m}(y)|_{y=0})\;[u_i^r X_i^m] e^{u_i W_{0,1}(X_i)}\frac{H(u_i,X_i)}{u_j\hbar\cS(u_j\hbar)}.
\end{equation}
If $H(u_i,X_i)/u_i$ is regular in $u_i$ (which is always the case for $n\geq 2$), then we have also
\begin{align}\label{eq:Uveered}
(U^\vee_iH)(w_i)&=\sum_{j=0}^\infty (\tfrac{1}{Q(X_i)}X_i\partial_{X_i})^j\sum_{r=0}^\infty
\tfrac{[v^j]L_r(-v,\hbar,W_{0,1}(X_i))}{Q(X_i)}
[u_i^r]\frac{H(u_i,X_i)}{u_j\hbar\cS(u_j\hbar)}\Bigm|_{X_i=X(w_i)},\\
L_{r}(v,\hbar,y) &\coloneqq \left(\partial_y+v\psi'(y)\right)^r e^{v\left(\frac{\cS(v\hbar\partial_y)}{\cS(\hbar\partial_y)}-1\right)\psi(y)},\\
Q(X_i)&\coloneqq\frac{X_i}{w_i}\frac{dw_i}{dX_i}=1+X_i\partial_{X_i}\log(\phi(W_{0,1}(X_i))),
\end{align}
where the change $X=X(w)$ is defined through an implicit equation
\begin{equation}\label{eq:Xw}
w=X\,\phi(W_{0,1}(X)).
\end{equation}
\end{theorem}
We give the proofs of the latter two theorems in Section~\ref{sec:DualProofs}.

\begin{corollary}\label{cor:rational} Assume that we are given rational functions $X(z)$, $w(z)$, $y(z)$, and $\phi(y)$ satisfying the identity
\begin{equation}
w(z)=X(z)\phi(y(z)).
\end{equation}
Consider the collections of functions $W_{g,n}$ and $W_{g,n}^\vee$ defined by~\eqref{eq:Wcorr}--\eqref{eq:Wveecorr} and expressed in $z_1,\dots,z_n$ through the substitutions $X_i=X(z_i)$ and $w_i=w(z_i)$, respectively. Then the functions $W_{g,n}$ are rational in $z$-coordinates for all $g$ and $n$ iff the functions $W^\vee_{g,n}$ are rational in $z$-coordinates for all $g$ and $n$.
\end{corollary}

\subsection{Special cases and examples} \label{sec:examples}
In this subsection we list explicit formulas for $W_{g,n}$ and $W^\vee_{g,n}$ in terms of one another for the special cases $(g,n)=(0,1)$ and $(g,n)=(0,2)$, and for regular examples $(g,n)=(0,3)$ and $(g,n)=(1,1)$.
\subsubsection{$(0,1)$ case} \label{sec:01case}

The case $(g,n)=(0,1)$ is a unique case when the sum over graphs has a constant term in the expansion in the $u$-variables and the relations~\eqref{eq:Ured},~\eqref{eq:Uveered} cannot be applied directly. However, the computations similar to those in~\cite{bychkov2021explicit} give in this case
\begin{align}
W_{0,1}(X)&=\sum_{r=0}^\infty\sum_{m=1}^\infty X^m\,\partial_y^r\phi(y)^m\bigm|_{y=0}[u^rw^m]\frac{e^{u\,W_{0,1}^\vee(w)}}{u}=W_{0,1}^\vee(w(X)),\\ X&=\frac{w(X)}{\phi(W_{0,1}^\vee(w(X)))},\\
W^\vee_{0,1}(w)&=\sum_{r=0}^\infty\sum_{m=1}^\infty X^m\,\partial_y^r\phi(y)^{-m}\bigm|_{y=0}[u^rw^m]\frac{e^{u\,W_{0,1}(X)}}{u}=W_{0,1}(X(w)),\\ w&=X(w)\;\phi(W_{0,1}(X(w))).\label{eq:Wvee01}
\end{align}
As a corollary, we conclude that \emph{the changes $w(X)$ and $X(w)$ defined by~\eqref{eq:wX} and~\eqref{eq:Xw} are inverse to one another and the functions $W_{0,1}(X)$ and $W_{0,1}^\vee(w)$ are identified through this change}. Indeed, we have
\begin{equation}
  X(w(X))
  \xlongequal{\!\eqref{eq:Xw}\!}\frac{w(X)}{\phi(W_{0,1}(X(w(X))))}
  \xlongequal{\!\eqref{eq:Wvee01}\!}\frac{w(X)}{\phi(W^\vee_{0,1}(w(X)))}
  \xlongequal{\!\eqref{eq:wX}\!}X,
\end{equation}
and similarly one shows $w(X(w))=w$.

\subsubsection{$(0,2)$ case}
Out of four connected multigraphs on two vertices only one gives a contribution for $g=0$, and we get
\begin{align}
W_{0,2}(X_1,X_2)+\tfrac{X_1X_2}{(X_1-X_2)^2}&=[\hbar^0]U_1U_2u_1u_2\Bigl(W^\vee_{0,2}(w_1,w_2)+\tfrac{w_1w_2}{(w_1-w_2)^2}\Bigr)\\
&=\tfrac1{Q^\vee(w_1)Q^\vee(w_2)}\Bigl(W^\vee_{0,2}(w_1,w_2)+\tfrac{w_1w_2}{(w_1-w_2)^2}\Bigr).\notag
\end{align}
The dual computation of $W_{0,2}^\vee$ yields
\begin{equation}
W^\vee_{0,2}(w_1,w_2)+\tfrac{w_1w_2}{(w_1-w_2)^2}=\tfrac1{Q(X_1)Q(X_2)}\Bigl(W_{0,2}(X_1,X_2)+\tfrac{X_1X_2}{(X_1-X_2)^2}\Bigr).
\end{equation}
The two computations agree since we have through the change between $X$ and $w$
\begin{equation}
Q(X)=\frac1{Q^\vee(w)}.
\end{equation}
These relations can be rewritten also in a more symmetric form
\begin{equation}
\omega^\vee_{0,2}=W^\vee_{0,2}(w_1,w_2)\tfrac{dw_1}{w_1}\tfrac{dw_2}{w_2}+\tfrac{dw_1dw_2}{(w_1-w_2)^2}
=W_{0,2}(X_1,X_2)\tfrac{dX_1}{X_1}\tfrac{dX_2}{X_2}+\tfrac{dX_1dX_2}{(X_1-X_2)^2}=\omega_{0,2}.
\end{equation}
\begin{remark} For the special case of maps/fully simple maps it is known that there exists a local coordinate $z$ on the line with coordinate $X$ or $w$ such that the above bidifferential becomes $\frac{dz_1dz_2}{(z_1-z_2)^2}$ in this coordinate (it is the global affine coordinate $z$ on the spectral curve $\Sigma ={\mathbb C} P^1$, see Section~\ref{sec:mapsdef}).	Computer experiments indicate that it should also be the case in general. However, the problem of finding this coordinate for the general case remains open, cf. Conjecture~\ref{conj:toprec}.	
\end{remark}

\subsubsection{$(0,3)$ case}
Set
\begin{align}
y_i&\coloneqq W_{0,1}(X_i)=W_{0,1}^\vee(w_i),\\
Q^\vee_i&\coloneqq Q^\vee(w_i)\\
Q_i&\coloneqq Q(X_i)=\frac{1}{Q^\vee_i}.
\end{align}
Then,
\begin{align}\label{eq:W03}
W_{0,3}(X_1,X_2,X_3)&=\tfrac{W^\vee_{0,3}(w_1,w_2,w_3)}{Q^\vee_1Q^\vee_2Q^\vee_3}+
\sum_{i=1}^3\tfrac1{Q^\vee_i}w_i\partial_{w_i}\tfrac{\psi'(y_i)
\prod_{j\ne i}\left(W^\vee_{0,2}(w_i,w_j)+\tfrac{w_iw_j}{(w_i-w_j)^2}\right)}{Q^\vee_1Q^\vee_2Q^\vee_3}\\
\label{eq:Wvee03}W^\vee_{0,3}(w_1,w_2,w_3)&=\tfrac{W_{0,3}(X_1,X_2,X_3)}{Q_1Q_2Q_3}-
\sum_{i=1}^3\tfrac1{Q_i}X_i\partial_{X_i}\tfrac{\psi'(y_i)
\prod_{j\ne i}\left(W_{0,2}(X_i,X_j)+\tfrac{X_iX_j}{(X_i-X_j)^2}\right)}{Q_1Q_2Q_3}
\end{align}
where we identify functions in $X_i$ and $w_i$ on the two sides of these equations through the change~\eqref{eq:wX} and~\eqref{eq:Xw}.

\subsubsection{$(1,1)$ case} Applying Theorem~\ref{th:WveeW} to this case we obtain, after some simplification,
\begin{align}\label{eq:W11}
W_{1,1}(X)=&\tfrac{1}{Q^\vee(w)}W^\vee_{1,1}(w)
 +(X\partial_X)^2\left(\tfrac{1}{Q^\vee(w)}\left(\tfrac{\psi''(y)}{24}w\partial_wy+\tfrac{\psi'(y)^2}{24}(w\partial_w)^2y\right)\right)\\
 &\qquad\notag
 {}+X\partial_X\left(\tfrac{1}{Q^\vee(w)}\left(\tfrac{\psi''(y)}{24}(w\partial_w)^2y
 -\tfrac{\psi'(y)}{24}+\tfrac{\psi'(y)}{2}W^\vee_{0,2}(w,w)\right)-\tfrac{\psi'(y)}{24}\right)\\
\label{eq:Wvee11}W^\vee_{1,1}(w)=&\tfrac{1}{Q(X)}W_{1,1}(X)
 +(w\partial_w)^2\left(\tfrac{1}{Q(X)}\left(-\tfrac{\psi''(y)}{24}X\partial_Xy+\tfrac{\psi'(y)^2}{24}(X\partial_X)^2y\right)\right)\\
 &\qquad\notag
 {}-w\partial_w\left(\tfrac{1}{Q(X)}\left(\tfrac{\psi''(y)}{24}(X\partial_X)^2y
 -\tfrac{\psi'(y)}{24}+\tfrac{\psi'(y)}{2}W_{0,2}(X,X)\right)-\tfrac{\psi'(y)}{24}\right)
\end{align}
where we identify similarly functions in $X$ and $w$ through the change~\eqref{eq:wX} and~\eqref{eq:Xw} and where we denote $y=W^\vee_{0,1}(w)=W_{0,1}(X)$. Note that we have through this change
\begin{equation}
X\partial_X=\tfrac{1}{Q^\vee(w)}w\partial_w,\qquad w\partial_w=\tfrac1{Q(X)}X\partial_X.
\end{equation}

\subsection{Duality proof} \label{sec:KPduality-end}
\label{sec:DualProofs}
\begin{proof}[Proof of Theorems~\ref{th:WWvee} and~\ref{th:WveeW}]
Consider $Z=e^F$ and $Z^\vee=e^{F^\vee}$ as elements of the (bosonic) Fock space:
\begin{align}
Z& =e^{\sum_{i=1}^\infty \frac{J_i t_i}{i\hbar}} \cD_\psi e^{\sum_{i=1}^\infty \frac{s_i J_{-i}}{i\hbar}}\,\big| 0 \big\rangle;
\\
Z^\vee& =\cD_\psi^{-1}e^{\sum_{i=1}^\infty \frac{J_i t_i}{i\hbar}} \cD_\psi e^{\sum_{i=1}^\infty \frac{s_i J_{-i}}{i\hbar}}\,\big| 0 \big\rangle.
\end{align}
In fact, an explicit structure of these expressions is not important for the proof. The only property that is used in the derivation of~\eqref{eq:Wgn} and~\eqref{eq:Wveegn} is the following relation between these functions.
\begin{equation}
Z=\cD_{\psi}Z^\vee,\qquad Z^\vee=\cD_{\psi}^{-1}Z=\cD_{-\psi}Z.
\end{equation}
Denote
\begin{align}
	\bJ_{m}&\coloneqq D_{\psi}^{-1}J_m D_{\psi},\\
	\bJ^\vee_{m}&\coloneqq D_{\psi}J_m D^{-1}_{\psi}.	
\end{align}
Consider first the disconnected $n$-point functions (or rather the extended versions $\widehat{W}_n^\bullet$ of these functions in the sense of~\cite[Section 4]{bychkov2021explicit}). By definition, we have
\begin{align}
\widehat W^\bullet_{n}&=\sum_{m_1,\dots,m_n}X_1^{m_1}\dots X_n^{m_n}
 \big\langle 0\! \bigm|J_{m_1}\dots J_{m_n} Z\\
&=\sum_{m_1,\dots,m_n}X_1^{m_1}\dots X_n^{m_n}
 \big\langle 0\! \bigm|J_{m_1}\dots J_{m_n} D_{\psi}Z^\vee\notag\\
&=\sum_{m_1,\dots,m_n}X_1^{m_1}\dots X_n^{m_n}
 \big\langle 0\! \bigm|\bJ_{m_1}\dots \bJ_{m_n} Z^\vee,\notag
\end{align}
where the summation runs over tuples of all possible integer values of $m_i$'s, both positive and negative. Then, \cite[Proposition~3.1]{bychkov2021explicit} states that
\begin{align}	
\bJ_{m}
&=\sum_{r=0}^\infty \bigl(\partial_y^r\phi_m(y)\bigm|_{y=0}\bigr)\;[u^rw^m]
\tfrac{e^{\sum_{k=1}^\infty u\hbar\cS(u\hbar k)J_{-k}w^{-k}}e^{\sum_{k=1}^\infty u\hbar\cS(u\hbar k)J_{k}w^{k}}}
{u\,\cS(u\hbar)}.
\end{align}
Similarly, for the dual case we get
\begin{align}
\widehat W^{\vee\bullet}_{n}&=\sum_{m_1,\dots,m_n}w_1^{m_1}\dots w_n^{m_n}
 \big\langle 0\! \bigm|\bJ^\vee_{m_1}\dots \bJ^\vee_{m_n} Z,\\
\bJ^\vee_{m}&=\sum_{r=0}^\infty \bigl(\partial_y^r\phi_{-m}(y)\bigm|_{y=0}\bigr)\;[u^rX^m]
\tfrac{e^{\sum_{k=1}^\infty u\hbar\cS(u\hbar k)J_{-k}X^{-k}}e^{\sum_{k=1}^\infty u\hbar\cS(u\hbar k)J_{k}X^{k}}}
{u\,\hbar\cS(u\hbar)}.\label{eq:cJvee}
\end{align}
Next, we substitute these expressions for $\bJ_m$ and commute the corresponding vertex operators moving positive $J$-operators to the right and negative ones to the left. Then, we compute
\begin{align}
\widehat W^\bullet_{n}&=\sum_{m_1,\dots,m_n}\sum_{r_1,\dots,r_n=0}^\infty \prod_{i=1}^n\bigl(\partial_y^{r_i}\phi_{m_i}(y)\bigm|_{y=0}\bigr)\;[\prod_{i=1}^n u_i^{r_i}w_i^{m_i}]\\
&\quad \prod_{i=1}^n\tfrac{1}{u_i\hbar\cS(u_i\hbar)}  \prod_{1\le k<\ell\le n}
e^{\hbar^2u_ku_\ell\cS(u_k\hbar w_k\partial_{w_k})\cS(u_\ell\hbar w_\ell\partial_{w_\ell})\frac{w_k w_\ell}{(w_k-w_\ell)^2}}\notag\\
&\qquad\times \big\langle 0\! \bigm|\prod_{i=1}^n e^{\sum_{k=1}^\infty u_i\hbar\cS(u_i\hbar k)J_{k}w_i^{k}} Z^\vee\notag\\
&=U_1\dots U_n \prod_{i=1}^n e^{-u_i W^\vee_{0,1}(w_i)}
\prod_{1\le k<\ell\le n}\!\!\!e^{\lwgt^{\vee\mathrm{corr}}_{k,\ell}}\;\;
\left(\prod_{i=1}^n e^{\tilde\delta^\vee_i}\right)e^{F^\vee}\Biggm|_{p=0}\notag
\\
&=U_1\dots U_n e^{-\sum_{i=1}^n u_i W^\vee_{0,1}(w_i)+
\sum_{1\le k<\ell\le n}\lwgt^{\vee\mathrm{corr}}_{k,\ell}}\;\;
  e^{\left(\prod_{i=1}^n e^{\tilde\delta^\vee_i}\right)F^\vee\bigm|_{p=0}}\notag
\\
&=U_1\dots U_n
e^{-\sum_{i=1}^n u_i W^\vee_{0,1}(w_i)+
\sum_{1\le k<\ell\le n}\lwgt^{\vee\mathrm{corr}}_{k,\ell}\;\;+\sum_{I\subset \{1,\dots,n\}}\lwgt^\vee_I}\notag
\\
&=U_1\dots U_n \prod_{I\subset \{1,\dots,n\}}(1+\wgt^\vee(I))\notag
\\
&=U_1\dots U_n \sum_\gamma\prod_{e\in E(\gamma)}\wgt^\vee(I)\notag,
\end{align}
where the last summation goes over all multigraphs $\gamma$ on $n$ numbered vertices, both connected and disconnected ones. The inclusion/exclusion procedure used in the passage from disconnected to connected $n$-point functions singles out exactly connected multigraphs. On the other hand, by~\cite{bychkov2021explicit}, the `connected extended $n$-point functions' are given by
\begin{equation}
\widehat W_{g,n}(X_1,\dots,X_n)=W_{g,n}(X_1,\dots,X_n)+\delta_{g,0}\delta_{n,2}\tfrac{X_1X_2}{(X_1-X_2)^2}.
\end{equation}
This completes the proof of~\eqref{eq:Wgn}. The proof of~\eqref{eq:Wveegn} is similar.

\end{proof}

\subsection{Maps vs fully simple maps}\label{sec:duality}

The enumeration of maps corresponds to the specialization $s_k=\delta_{k,2}$ and $\phi(y)=1+y$. We have in this case
\begin{equation}	
\phi(y)=1+y=\frac{w}{X}.
\end{equation}
Therefore, we have
\begin{align}
\omega_{0,1}&=y\frac{d X}{X}=X\,(1+y)\frac{d X}{X^2}-\frac{d X}{X}=-w\,d(X^{-1})-\frac{dX}{X},\\
\omega^\vee_{0,1}&=y\frac{d w}{w}=\frac{1+y}{w}d w-\frac{d w}{w}=X^{-1}\,d w-\frac{d w}{w}
\end{align}
The topological recursion relations make use of the odd part of $\omega_{0,1}$ with respect to the involution only. It follows that adding even summands to this form does not change the result of the recursion. Therefore, we can set in this case equally
\begin{equation}
	\omega_{0,1}=-w\,d(X^{-1}),\qquad \omega^\vee_{0,1}=X^{-1}\,d w,
\end{equation}
and this will not affect the computation of the higher differentials. Therefore, our duality statement is equivalent in this case to the `$(x\leftrightarrow y)$-symmetry' conjecture of~\cite{borot2018simple}: we use $X^{-1}$ and $w$ as the conventional $x$ and $y$ coordinates of topological recursion. Note that in the general case our formulation of ordinary vs fully simple duality is not reduced to the $x\leftrightarrow y$ duality.

\subsection{Beyond KP integrability}\label{sec:BeyondKP} The generalised ordinary vs fully simple duality theory developed in Sections~\ref{sec:KPDuality-beginning}-\ref{sec:KPduality-end} can be applied (with literally the same proofs) to more general VEVs.
Let $\psi=\sum_{k=1}^\infty\sum_{m=0}^\infty c_{k,m}y^k\hbar^{2m}$ be a formal series in $y$ and $\hbar^2$, and $t_{g;k_1,\dots,k_m}$ and $s_{g;k_1,\dots,k_m}$ (for $g\ge0$, $m\ge 1$, $k_i\ge1$) be some numbers symmetric in the indices $k$. Here we no longer assume that only a finite number of the coefficients $c$, $t$, and $s$ are non-zero (as we did at the beginning of this section). Define
\begin{align}
	&\Mf^\bullet_n(X_1,\dots,X_n) \coloneqq \\ \nonumber
	&\quad = \big\langle 0 \big|  \left(\prod_{i=1}^n \sum_{\ell_i=1}^\infty \dfrac{J_{\ell_i}}{\ell_i} X_i^{\ell_i}\right) \exp \left({\sum_{g=0}^\infty \sum_{m=1}^{\infty} \frac{\hbar^{2g-2+m}}{m!} \sum_{k_1,\dots,k_m=1}^\infty t_{g;k_1,\dots,k_m} \prod_{j=1}^m \frac{J_{k_j}}{k_j} }\right)   \\ \notag
	& \qquad \qquad
	\times\cD_{\psi}\exp \left({\sum_{g=0}^\infty \sum_{m=1}^{\infty} \frac{\hbar^{2g-2+m}}{m!} \sum_{k_1,\dots,k_m=1}^\infty s_{g;k_1,\dots,k_m} \prod_{j=1}^m \frac{J_{-k_j}}{k_j} }\right) \big| 0 \big\rangle\\ \nonumber
	&\quad =\big\langle 0 \big|  \left(\prod_{i=1}^n \sum_{\ell_i=1}^\infty \dfrac{J_{\ell_i}}{\ell_i} X_i^{\ell_i}\right)  \cD_{\psi}
	\exp \left({\sum_{g=0}^\infty \sum_{m=1}^{\infty} \frac{\hbar^{2g-2+m}}{m!} \sum_{k_1,\dots,k_m=1}^\infty \widetilde{s}_{g;k_1,\dots,k_m} \prod_{j=1}^m \frac{J_{-k_j}}{k_j} }\right)
	\big| 0 \big\rangle, \\
	&\Mf^{\vee\bullet}_n(w_1,\dots,w_n) \coloneqq \\  \nonumber
	&\quad = \big\langle 0 \big|  \left(\prod_{i=1}^n \sum_{\ell_i=1}^\infty \dfrac{J_{\ell_i}}{\ell_i} w_i^{\ell_i}\right) \cD_{\psi}^{-1}  \exp \left({\sum_{g=0}^\infty \sum_{m=1}^{\infty} \frac{\hbar^{2g-2+m}}{m!} \sum_{k_1,\dots,k_m=1}^\infty t_{g;k_1,\dots,k_m} \prod_{j=1}^m \frac{J_{k_j}}{k_j} }\right)
	\\ \notag & \qquad \qquad \times \cD_{\psi}
	\exp \left({\sum_{g=0}^\infty \sum_{m=1}^{\infty} \frac{\hbar^{2g-2+m}}{m!} \sum_{k_1,\dots,k_m=1}^\infty s_{g;k_1,\dots,k_m} \prod_{j=1}^m \frac{J_{-k_j}}{k_j} }\right)
	\Big)\big| 0 \big\rangle
	\\ \nonumber
	&\quad =\big\langle 0 \big|  \left(\prod_{i=1}^n \sum_{\ell_i=1}^\infty \dfrac{J_{\ell_i}}{\ell_i} w_i^{\ell_i}\right)
	\exp \left({\sum_{g=0}^\infty \sum_{m=1}^{\infty} \frac{\hbar^{2g-2+m}}{m!} \sum_{k_1,\dots,k_m=1}^\infty \widetilde{s}_{g;k_1,\dots,k_m} \prod_{j=1}^m \frac{J_{-k_j}}{k_j} }\right)
	\big| 0 \big\rangle,
\end{align}
with $\Mf_n(X_1,\dots,X_n)$ and $\Mf^\vee_n(X_1,\dots,X_n)$ being the connected counterparts, respectively. Do note that while we list the variants of these formulas containing the $t$-parameters (in order for the maps/fully simple maps case to be clearly seen as a special case here), these $t$-parameters do not represent independent degrees of freedom w.~r.~t. the $s$-parameters, which is visible from the $\widetilde{s}$-rewritten formulas.

This more general case includes, in particular, the stuffed maps~\cite{BorotStuffed,BorotSha-Blobbed,borot2018simple} (for $\psi
=\log(1+y)$ and general values for the parameters $s_{g;k_1,\dots,k_m}$), and some pieces of the generalised ordinary vs fully simple duality statements in this case are already mentioned in~\cite[Section 7]{borot2018simple}. Using the results of~\cite[Section 5]{borot2019relating}, our Lemma~\ref{lem:FSviaHurwitz} also proves that $\Mf_n^\vee$ enumerates the corresponding fully simple objects (stuffed maps and hypermaps).

\subsection{General duality for \texorpdfstring{$\Mf$}{H}-functions} \label{sec:general duality}
In this section we formulate the refined form of the duality formulated in terms of the functions $\Mf_{g,n}\coloneqq[\hbar^{2g-2+n}]\Mf_n$ and $\Mf^\vee_{g,n}\coloneqq[\hbar^{2g-2+n}]\Mf^\vee_n$ (we use here the same notation as in \eqref{eq:Mveedef1}, \eqref{eq:Mveedef} for the generalised n-point function) defined in Section~\ref{sec:BeyondKP}. It takes the following form, where $W_{0,1}(X)=X\partial_X \Mf_{0,1}(X)$ and $\psi_0 \coloneqq \psi|_{\hbar=0}$:
\begin{theorem}\label{prop:Mveeform}
	For $n\geq 3$:
	\begin{align}\label{eq:Mveeform}
		&\Mf_{g,n}^\vee(w_1,\dots,w_n)=[\hbar^{2g-2+n}] \sum\limits_{\gamma\in\tilde\Gamma_n} \prod\limits_{v_i\in\mathcal{I}_\gamma} \overline{U}^\vee_i \prod\limits_{e\in E_\gamma\setminus\mathcal K_\gamma}\wgt
		(e)
		\\&\quad\times \nonumber
		\prod\limits_{e\in\mathcal{K}_\gamma}\left( \overline{U}^\vee_{i_1(e)}\cdots\overline{U}^\vee_{i_{l(e)}(e)} \wgt
		(e) +  \widetilde{\wgt}_{i_1(e)\dots i_{l(e)}(e)}
		(e)
		\right)
		+\mathrm{const},
	\end{align}
	where $\tilde\Gamma_n$ is the set of all connected multigraphs on $n$ vertices $v_1,\ldots,v_n$, $E_\gamma$ is the set of multiedges of a graph $\gamma$, $\mathcal{I}_\gamma$ is the subset of vertices of valency $\geq 2$, and $\mathcal{K}_\gamma$ is the subset of multiedges $e$ with $l(e) > 0$ ends $v_{i_1(e)},\dots,v_{i_{l(e)}(e)}$ of valency $1$, and where 
	\begin{align}	
		\overline U^\vee_i\, f &\coloneqq\sum_{j=1}^\infty (\tfrac{1}{Q(X_i)}X_i\partial_{X_i})^{j-1}\sum_{r=0}^\infty
		\tfrac{[v^j]L_r(-v,\hbar,W_{0,1}(X_i))}{Q(X_i)}
		[u_i^r]\frac{f}{u_j\hbar\cS(u_j\hbar)}\Bigm|_{X_i=X(w_i)},\\
		L_{r}(v,\hbar,y) &\coloneqq \left(\partial_y+v\psi_0'(y)\right)^r e^{v\left(\frac{\cS(v\hbar\partial_y)}{\cS(\hbar\partial_y)}\psi-\psi_0\right)},\\
		Q(X_i)&\coloneqq\frac{X_i}{w_i}\frac{dw_i}{dX_i}=1+X_i\dfrac{d}{dX_i}\psi_0(W_{0,1}(X_i)),
	\end{align}
	and
	\begin{align}
		\wgt
		(j_1,\dots,j_k)&\coloneqq
		e^{\lwgt_{j_1,\dots,j_k}+\delta_{k,2}\lwgt^{\mathrm{corr}}_{j_1,j_2}-\delta_{k,1}u_{j_1} W_{0,1}(X_{j_1})}-1,
		\\
		\lwgt_{j_1,\dots,j_k}&\coloneqq\sum_{r_1,\dots,r_k=1}^\infty\tfrac{1}{r_1!\cdots r_k!}\lwgtpart_{\underbrace{\scriptstyle j_1,\dots,j_1}_{r_1},\dots,\underbrace{\scriptstyle j_k,\dots,j_k}_{r_k}},\\
		\lwgtpart_{j_1,\dots,j_m}
		&\coloneqq\Biggl(\prod_{i=1}^m u_{j_i}\hbar\,\cS(u_{j_i}\,\hbar X_{j_i}\partial_{X_{j_i}}) X_{j_1}\partial_{X_{j_i}}\Biggr)
		\sum_{g=0}^\infty \hbar^{2g-2+m}\Mf_{g,m}(X_{j_1},\dots,X_{j_m}),\\
		\lwgt^{\mathrm{corr}}_{i,j}&\coloneqq\hbar^2 u_{i}u_{j}\cS(u_{i}\hbar X_{i}\partial_{X_{i}})\cS(u_{j}\hbar X_{j}\partial_{X_{j}})\frac{X_{i} X_{j}}{(X_{i}-X_{j})^2},
	\end{align}	
	and
	\begin{align}
		\widetilde{\wgt}_{i_1\dots i_l}
		(j_1,\dots,j_k)&\coloneqq\delta_{k,2}\hbar u_{\bar i_1} \cS(u_{\bar i_1}\hbar X_{\bar i_1}\partial_{X_{\bar i_1}})\frac{X_{i_1}}{X_{\bar i_1}-X_{i_1}}+\widetilde{\lwgt}_{\bar i_1\dots\bar i_m}^{i_1\dots i_{l}},
	\end{align}	
	where $m=k-l$ and $\{\bar i_1,\dots,\bar i_m\} = \{j_1,\dots,j_k\} \setminus \{i_1,\dots,i_{l}\}$, and
	\begin{align}
		\widetilde{\lwgt}_{\emptyset}^{i_1\dots i_{l}}&\coloneqq\widetilde{\lwgtpart}^{i_1\dots i_l}_{\emptyset}, \qquad m=0\\
		\widetilde{\lwgt}_{j_1\dots j_m}^{i_1\dots i_{l}}&\coloneqq
		\sum_{r_1,\dots,r_m=1}^\infty \dfrac{1}{r_1!\cdots r_m!} \widetilde{\lwgtpart}^{i_1\dots i_l}_{\underbrace{\scriptstyle j_1,\dots,j_1}_{r_1},\dots,\underbrace{\scriptstyle j_m,\dots,j_m}_{r_m}}, \qquad m\geq 1\\
		\widetilde{\lwgtpart}^{i_1\dots i_l}_{j_1\dots j_q} &\coloneqq 	\hbar^l \Biggl(\prod_{s=1}^q u_{j_s}\hbar\,\cS(u_{j_s}\,\hbar X_{j_s}\partial_{X_{j_s}}) X_{j_s} \partial_{X_{j_s}}\Biggr) \\ \notag & \qquad
		{\sum_{g=0}^\infty}^\prime \hbar^{2g-2+l+m}\Mf_{g,l+q}(X_{i_1},\dots,X_{i_l},X_{j_1},\dots,X_{j_q}),
	\end{align}
where in the last sum the prime means that for $(l,q)=(1,0)$ the sum runs from $g=1$ instead of $g=0$.

	For $n=2$ and $g>0$ we have:
	\begin{align}\label{eq:Mveeform2}
		\Mf^\vee_{g,2}& =\eqref{eq:Mveeform}^\prime+ [\hbar^{2g}]\Bigg(\overline{U}^\vee_1\overline{U}^\vee_2 \wgt(1,2)
		+ \overline{U}^\vee_1\widetilde{\wgt}_{2}(1,2) +\overline{U}^\vee_2\widetilde{\wgt}_{1}(1,2)\Bigg)
		+\mathrm{const},
	\end{align}
where by $\eqref{eq:Mveeform}^\prime$ we mean the expression one obtains from \eqref{eq:Mveeform} for this $n=2$ case, but instead of the summation over all 4 graphs there are in this case, we sum over only 3 of them, excluding the graph $\gamma$ with $E(\gamma)=\{(1,2)\}$.

	For $n=1$ and $g>0$ we have:
	\begin{align} \label{eq:Mveeform1}
		\Mf^\vee_{g,1}& =[\hbar^{2g-1}]\Bigg(\overline{U}^\vee_1 \wgt(1) +\widetilde{\wgt}_{1}(1) \\ \notag
		& \quad +\sum_{j=1}^\infty \Big(\tfrac 1{Q(X_1)}X_1\partial_{X_1}\Big)^{j-1}  \tfrac{[v^{j+1}]L_0(-v,\hbar,W_{0,1}(X_1))}{\hbar Q(X_1)}  X_1\partial_{X_1}W_{0,1}(X_1)
		\\ \notag & \quad
		- 
		\int_{0}^{W_{0,1}(X_1)}\left(\dfrac{1}{\cS(\hbar\partial_y)}\psi(y,\hbar)-\psi_0(y)\right)dy\Bigg)
		+\mathrm{const}.
	\end{align}
\end{theorem}
\begin{proof}
	Theorem~\ref{th:WveeW} holds for the general $W^\vee_{g,n}$ and $W_{g,n}$ corresponding to $\Mf^\vee_{g,n}$ and $\Mf_{g,n}$ of Section~\ref{sec:BeyondKP} (in place of KP-type $W^\vee_{g,n}$ and $W_{g,n}$) in exactly the same form, and the proofs are completely analogous. Passing from $W^\vee_{g,n}$ to $\Mf^\vee_{g,n}$ is then done analogously to how a similar thing was done in \cite[Theorem~5.3 and Section~6]{bychkov2021explicit}.
\end{proof}

\begin{remark}
	Analogously, $\Mf_{g,n}$'s can be expressed in terms of $\Mf^\vee_{g,n}$'s in similar way to Theorem~\ref{prop:Mveeform}. The difference is as Theorem~\ref{th:WveeW} differs from Theorem~\ref{th:WWvee}. One just has to take all formulas of Theorem~\ref{prop:Mveeform} and then to replace $\psi$ with $-\psi$ and all $\vee$-d functions with non-$\vee$-d ones and vice versa everywhere; this follows from the fact that $\cD_{\psi}^{-1}=\cD_{-\psi}$. We do not list explicit formulas here for brevity.
\end{remark}

\begin{example}
In the case $(g,n)=(0,3)$, applying $(X_1X_2X_3\partial_{X_1}\partial_{X_2}\partial_{X_3})^{-1}$ to both sides of formulas for general $W_{0,3}$ and  $(w_1w_2w_3\partial_{w_1}\partial_{w_2}\partial_{w_3})^{-1}$ to both sides of formulas for general $W_{0,3}^\vee$ (formulas for general $W_{0,3}$ and $W_{0,3}^\vee$ are analogous to the ones in Section \ref{sec:examples}), we obtain
\begin{align}
\Mf_{0,3}(X_1,X_2,X_3)-\Mf^\vee_{0,3}(w_1,w_2,w_3)=&
\sum_{i=1}^3\tfrac{\psi_0'(y_i)}{Q_i}
\prod_{j\ne i}\left(D_{w_i}\Mf^\vee_{0,2}(w_i,w_j)+\tfrac{w_j}{w_i-w_j}\right)+\psi_0'(0)\\
=&\sum_{i=1}^3\tfrac{\psi_0'(y_i)}{Q^\vee_i}
\prod_{j\ne i}\left(D_{X_i}\Mf_{0,2}(X_i,X_j)+\tfrac{X_j}{X_i-X_j}\right)+\psi_0'(0).\notag
\end{align}
Here
\begin{align}
\psi_k(y)&=[\hbar^{2k}]\log\phi(y)\\
D_X&=X\partial_{X}=\tfrac{1}{Q^\vee(w)}w\partial_w,& D_w&=w\partial_w=\tfrac{1}{Q(X)}X\partial_X,
\end{align}

The integration constant $\psi_0'(0)$ is computed from the condition that both sides should vanish at the origin, see~\cite{bychkov2021explicit} for more details on the computation of this kind of constants.

Analogously, in the case $(g,n)=(1,1)$ we obtain
\begin{align}
\Mf_{1,1}-\Mf^\vee_{1,1}={}&
\tfrac{1}{Q^\vee}\left(\tfrac{\psi_0''(y)}{24}D_w^2y
 -\tfrac{\psi_0'(y)}{24}+\tfrac{\psi_0'(y)}{2}W^\vee_{0,2}(w,w)\right)\\
 &\notag
 {}+D_X\left(\tfrac{1}{Q^\vee}\left(\tfrac{\psi_0''(y)}{24}D_wy+\tfrac{\psi_0'(y)^2}{24}D_w^2y\right)\right)
 -\tfrac{\psi_0'(y)}{24}-\int_0^y\psi_1(u)\,du+\tfrac{\psi_0'(0)}{12}
\\
={}&\notag
\tfrac{1}{Q}\left(\tfrac{\psi_0''(y)}{24}D_X^2y
 -\tfrac{\psi_0'(y)}{24}+\tfrac{\psi_0'(y)}{2}W_{0,2}(X,X)\right)\\
 &\notag
 {}+D_w\left(\tfrac{1}{Q}\left(\tfrac{\psi_0''(y)}{24}D_Xy-\tfrac{\psi_0'(y)^2}{24}D_X^2y\right)\right)
 -\tfrac{\psi_0'(y)}{24}-\int_0^y\psi_1(u)\,du+\tfrac{\psi_0'(0)}{12}.
\end{align}
\end{example}


\section{Proof of topological recursion for fully simple maps}\label{sec:TRforFS}

In this section we prove Theorem \ref{th:BGF}. Starting from this moment and up to the end of the paper we work only with a special case of fully simple maps, though, when possible, we still write more general formulas in the intermediate computations. So, from now on, as in Section~\ref{sec:mapsdef}, we assume that $\phi(y)=1+y$, $\psi(y)=\log(1+y)$, the functions $W_{g,n}(X_1,\dots,X_n)$ are the corresponding $n$-point functions for maps, and we use that these functions are rational functions on $\mathbb{C}P^1$ in the coordinates $z_1,\dots,z_n$, $X_i=X(z_i)$, with the coefficients depending on $t$-parameters, and with known singularities. Furthermore, they satisfy the topological recursion.

An immediate observation that follows from Theorem~\ref{th:WveeW} and topological recursion for the ordinary maps is that $W^\vee_{g,n}(w_1,\dots,w_n)$ are also rational functions on $\mathbb{C}P^1$ in the
coordinates $z_1,\dots,z_n$, $w_i=w(z_i)$ (cf. Corollary~\ref{cor:rational}).

\subsection{Loop equations}\label{sec:LoopEquations}
In this section we prove the loop equations (the blobbed topological recursion) for the fully simple maps. The way we present it is a bit streamlined and shortened version of the argument in~\cite[Section 2]{bychkov2021explicit}, and it is heavily based on \emph{op. cit.}, so we skip here many subtleties explained there in detail. We use the following notation for operators on the bosonic Fock space:
\begin{align}
	\sum_{k,l\in\Z+\frac12} x^ly^{-k} \hat E_{k,l} & \coloneqq \frac{x^{\frac 12} y^{\frac 12} \big( e^{\sum_{i=1}^\infty (y^{-i}-x^{-i}) \frac{p_i}i } e^{\sum_{i=1}^\infty (x^{i}-y^{i})  \partial_{p_i}} -1  \big) }{x-y}; \\
	\cE_0(u)&\coloneqq\sum_{k\in\Z+\frac12}e^{u\,k}\hat E_{k,k};\\
	J(w)& \coloneqq \sum_{m\in\Z} J_mw^m.
\end{align}

\begin{definition}
	\label{def:cW}
	Let
	\begin{align} \label{eq:VEVdefinitionOfcWgn}
		& \cW_{n}^{\vee\bullet}(w_1;w_{\llbracket n \rrbracket \setminus 1};u)\coloneqq
		\VEV{ \left(\sum_{m=1}^\infty \frac{w_1^m J_m}{m\hbar}\right) \cE_0(\hbar\,u)
			\left(\prod_{i=2}^nJ(w_i)\right)  \cD_\psi^{-1} e^{\sum_{i=1}^\infty \frac{J_i t_i}{i\hbar}}
			\cD_\psi e^{\sum_{j=1}^\infty \frac{s_j J_{-j}}{j\hbar} }}.
	\end{align}
	Using inclusion-exclusion procedure as described in Section \ref{sec:VEVs}, we also consider the connected versions of these VEVs denoted by $\cW^\vee_{n}$. Let
	$\cW_{g,n}^{\vee}\coloneqq [\hbar^{2g-2+n}]\cW_{n}^{\vee}$.
\end{definition}


\begin{lemma}[See~\cite{bychkov2020topological}] We have:
	\begin{align}
		[u^0]\cW^\vee_{g,n}(w_1;w_{\llbracket n \rrbracket \setminus 1};u)& =0; \\
		[u^1]\cW^\vee_{g,n}(w_1;w_{\llbracket n \rrbracket \setminus 1};u)& =W^\vee_{g,n}(w_{\llbracket n \rrbracket}); \\
		[u^2]\cW^\vee_{g,n}(w_1;w_{\llbracket n \rrbracket \setminus 1};u)& =
		\frac 12 W^\vee_{g-1,n+1}(w_1,w_1,w_{\llbracket n \rrbracket \setminus 1}) \\ \notag
		& \quad +
		\frac 12 \sum_{\substack{g_1+g_2=g\\I\sqcup J=\llbracket n \rrbracket \setminus 1}}
		W^\vee_{g_1,|I|+1}(w_1,w_{I})
		W^\vee_{g_2,|J|+1}(w_1,w_{J})
		\\ \notag &
		\quad
		+ \sum_{i\in\llbracket n\rrbracket \setminus 1} \frac{w_1w_i}{(w_1-w_i)^2} W^\vee_{g,n-1}(w_1,w_{\llbracket n \rrbracket \setminus 1,i}).
	\end{align}
\end{lemma}

So, in order to proof the linear and quadratic loop equations it is sufficient to prove that the corresponding coefficients of $\cW^\vee_{g,n}$ are rational functions in $z_1$, $w_1=w(z_1)$, whose polar parts at the critical points of $w$ are odd with respect to the respective deck transformations. We can show this via deriving explicit closed formulas for $\cW^\vee_{g,n}$. To this end, we obtain the following straightforward generalisation of Theorem~\ref{th:WveeW}.

\begin{lemma} \label{lem:curlyWexplicit} For $g\geq 0$, $n\geq 2$, $2g-2+n>0$, we have:
	\begin{equation}
		\label{eq:cWvee}
		\cW^\vee_{g,n}(w_1;w_{\llbracket n \rrbracket \setminus 1};u) =[\hbar^{2g-2+n}]U^{\vee}_n\ldots U^{\vee}_2\widetilde{U}^{\vee}_1\sum\limits_{\gamma\in\tilde{\Gamma}_n}\prod\limits_{e\in E_\gamma} \mathrm{wgt}_\gamma(e),
	\end{equation}
	where $U^\vee_i$, $i=2,\dots,n$, are given by Equation~\eqref{eq:Uveered}, and the transformation $\widetilde{U}^\vee_1$ takes a function $H(u_1,X_1)$ to the function $\widetilde{U}^\vee_1 H$ in $w_1$ given by
	\begin{align}
		(\widetilde U^\vee_1H)(w_1)& =
		\sum_{j=0}^\infty \Big(\frac{1}{Q(X_1)}X_1\partial_{X_1}\Big)^j [v^j]
		\\ \notag & \quad \sum_{r=0}^\infty
		\frac{L_r(-v,\hbar,W_{0,1}(X_1))uS(vu\hbar)e^{uW_{0,1}(X_1)}}{Q(X_1)}
		[u_1^r]\frac{H(u_1,X_1)}{u_1\hbar\cS(u_1\hbar)}\Bigm|_{X_1=X(w_1)},
	\end{align}
	where $w=X\,\phi(W_{0,1}(X)) = X(1+W_{0,1}(X))$.
	
	For $(g,n)=(0,2)$ we have:
	\begin{equation}
		\cW^\vee_{0,2} = \frac{ue^{u\,W_{0,1}(X_1)}}{Q_1Q_2}\Bigg(W_{0,2}(X_1,X_2)+\frac{X_1X_2}{(X_1-X_2)^2}\Bigg).
	\end{equation}
	
	For $n=1$, $g>0$ we have:
	\begin{align}
		& \cW^\vee_{g,1}  = [\hbar^{2g-1}]\widetilde{U}^{\vee}_1\sum\limits_{\gamma\in\tilde{\Gamma}_1}\prod\limits_{e\in E_\gamma} \mathrm{wgt}_\gamma(e)+
		\\ \notag &
		\sum_{j=0}^\infty \Big(\frac{1}{Q(X_1)}X_1\partial_{X_1}\Big)^j [v^{j+1}]
		\frac{L_0(-v,\hbar,W_{0,1}(X_1))uS(vu\hbar)e^{uW_{0,1}(X_1)}X_1\partial_{X_1} W_{0,1}(X_1)}{\hbar Q(X_1)}
		\Bigm|_{X_1=X(w_1)}
	\end{align}
	(of course, $|\tilde{\Gamma}_1|=1$ and for the only $\gamma\in\tilde{\Gamma}_1$ we have $|E_\gamma|=1$, but it is still convenient to write the expression this way).
	
	Finally, for $(g,n)=(0,1)$ we have:
	\begin{equation}	
		\cW^\vee_{0,1} = e^{u\,W_{0,1}(X_1)} -1.
	\end{equation}
\end{lemma}

\begin{proof}
	We compute $\cW^\vee_{g,n}$ in a closed form applying the techniques developed in~\cite{bychkov2020topological}, see also Section~\ref{sec:DualProofs}. From~\cite[Proof of Proposition~2.26]{bychkov2020topological} we have
	\begin{align}
		& \cW_{n}^{\vee\bullet}(w_1;w_{\llbracket n \rrbracket \setminus 1};u)\\ \notag
		&=[\epsilon^1]\VEV{ \left(\sum_{m=1}^\infty \frac{w_1^m J_m}{m\hbar}\right) e^{\epsilon \cE_0(\hbar\,u)}
			\left(\prod_{i=2}^nJ(w_i)\right)  \cD_\psi^{-1} e^{\sum_{i=1}^\infty \frac{J_i t_i}{i\hbar}}
			\cD_\psi e^{\sum_{j=1}^\infty \frac{s_j J_{-j}}{j\hbar} }}\\ \notag
		&=[\epsilon^1]\big\langle 0 \big|
		\cD_{\tilde \phi}\left(\sum_{m=1}^\infty \frac{w_1^m}{m\hbar } J_m\right)\cD_{\tilde \phi^{-1}}
		\left(\prod_{i=2}^n \cD_\psi J(w_i)\cD_{\phi^{-1}}\right)
		\\ \notag & \quad
		\exp\Bigg(\sum_{g\geq 0, k\geq 1} \frac{ \hbar^{2g-2+k}}{k!} \sum_{p_1,\dots,p_k=1}^\infty \prod_{i=1}^k \frac{J_{-p_i}}{p_i} \Big[\prod_{i=1}^k X_i^{p_i}\Big]W_{g,n}(X_1,\dots,X_k) \Bigg)
		\big| 0 \big\rangle    ,
	\end{align}
	where $\log\tilde\phi(y) = \psi(y)-\epsilon u \hbar S(\hbar \partial_y) e^{uy}$. This allows to apply directly the same computations as in Section~\ref{sec:DualProofs}, which leads to the formulas stated in this lemma.
\end{proof}

\begin{corollary}\label{cor:lqloop} The linear and quadratic loop equation hold for fully simple maps.
\end{corollary}

\begin{proof} Indeed, the formulas that we obtain in Lemma~\ref{lem:curlyWexplicit} manifestly show that the coefficients of the expansion in $u$ of $\cW^\vee_{g,n}$ are rational functions on $\mathbb{C}P^1$, whose singularities in $z_1$ at the critical points of $w_1=w(z_1)$ are generated by a finite iterative application of $w_1\partial_{w_1}$ to rational functions locally holomorphic at the respective critical points. By~\cite[Proposition 2.9 and 2.12]{bychkov2020topological}, this property implies the quadratic and linear loop equations.
\end{proof}


\subsection{Projection property}
We are going to closely follow \cite[Sections~3 \& 4]{bychkov2020topological}, though there are some new effects that we have to take into account. Recall from \cite{bychkov2020topological} the definition of the space $\Theta_n$:
\begin{definition}
	The space $\Theta$ (which we denote by $\Theta_n$, $n\geq 1$, when we want to stress the number of variables) is  defined as the linear span of functions of the form
	$\prod_{i=1}^n f_i(z_i)$,
	where each $f_i(z_i)$
	\begin{itemize}
		\item is a rational function on the Riemann sphere;
		\item has poles only at the zeroes of $dw(z)$ (which we denote by  $p_1,\ldots,p_N$);
		\item its principal part at $p_k$, $k=1,\dots,N$, is odd with respect to the corresponding deck transformation, that is, $f_i(z_i)+f_i(\sigma_k(z_i))$ is holomorphic at $z_i\to p_k$.
	\end{itemize}
\end{definition}

In order to prove the projection property, we need to prove that for $2g-2+n>0$ we have $\Mf^\vee_{g,n}\in\Theta_n$~\cite[Proposition~3.9]{bychkov2020topological}. Consider the statement of Theorem~\ref{prop:Mveeform} specified to fully simple maps expressions in terms of maps, as defined in Section~\ref{sec:mapsdef}. From Equations~\eqref{eq:Mveeform}, \eqref{eq:Mveeform2}, and~\eqref{eq:Mveeform1} specified to this case we see that for $2g-2+n>0$ every $\Mf^\vee_{g,n}$ is a rational function in $z_1,\dots,z_n$. Indeed, all $\Mf_{g,n}$ (without "${\,}^\vee$", the $n$-point functions for maps) entering the respective formulas for $\Mf^\vee_{g,n}$ are rational functions in $z_1,\dots,z_n$ (as implied by the topological recursion for maps), except for the $\Mf_{0,1}$ and $\Mf_{0,2}$ cases. But $\Mf_{0,1}$ itself does not enter the formulas~\eqref{eq:Mveeform}, \eqref{eq:Mveeform2}, and~\eqref{eq:Mveeform1} for $\Mf^\vee_{g,n}$, we only use $W_{0,1}$ , which is rational, and $\Mf_{0,2}$ only appears with at least one derivative taken, which is also a global rational function.
From the structure of the aforementioned formulas of Theorem~\ref{prop:Mveeform} we also immediately see that $\Mf^\vee_{g,n}$ can have poles only at the following points:
\begin{itemize}
	\item At the zeroes of $Q(X(z))$ which are the same as the zeroes of $dw(z)$. These are the expected poles of $\Mf^\vee_{g,n}$ which is fine for our purpose of proving $\Mf^\vee_{g,n}\in\Theta_n$.
	\item At $z=\pm 1$, since the functions $\Mf_{g,n}$ (for $2g-2+n>0$), entering the expression for $\wgt(e)$, have poles at these points (and only at these points, as it follows from the topological recursion for maps).
	\item At $z$ such that $y(z)=-1$, since they appear in $L_r$.
	\item At $z=0$ and $z=\infty$, since the derivatives of $y(z)=W_{0,1}(X(z))$ can poles at these points.
\end{itemize}
Let us prove that all these poles, except the ones at the zeroes of $Q(X(z))$, actually get cancelled.
\begin{lemma} \label{prop:nopmone}For $2g-2+n>0$,
$	\Mf_{g,n}^\vee(w(z_1),\dots,w(z_n))$ has no poles at $z_1=\pm 1$ and $z_1=0$.
\end{lemma}
\begin{proof}
	As mentioned above, Theorem~\ref{prop:Mveeform} implies that for $2g-2+n>0$ every $n$-point function $\Mf^\vee_{g,n}(w(z_1),\dots,w(z_n))$ is a rational function in $z_1,\dots,z_n$, also depending on the $t$-variables. Consider the expansion of $\Mf^\vee_{g,n}$ in the $t$-variables. If $\Mf^\vee_{g,n}$ had a pole at, say, $z_1=1$, then at least one of the coefficients $[t_{a_1}\cdots t_{a_k}]\Mf^\vee_{g,n}$ would have such a pole as well (as a function in $z_1,\dots,z_n$). This means that it is enough to prove that $\forall k\in \mathbb{Z}_{\geq 0} \; \forall a_1,\dots,a_k \in \mathbb{Z}_{>0}$ the following expression
	\begin{equation}
		 \left.\left(\dfrac{\partial^k}{\partial_{t_{a_1}}\dots\partial_{t_{a_k}}}\Mf^\vee_{g,n}(w(z_1),\dots,w(z_n))\right)\right|_{t=0}
	\end{equation}
has no poles at $z_1=1$ (and similarly for the poles at $z_1=-1$ and $z_1=0$).
	
	In fact, let us prove that $(\partial_{t_{a_1}}\cdots\partial_{t_{a_k}}\Mf^\vee_{g,n})|_{t=0}$ is a polynomial, and thus has no poles at all, other than potentially at infinity.
	
	We have 	
	\begin{align}
		&\Mf^{\vee
		}_{g,n}(w_1,\dots,w_n) =
		[\hbar^{2g-2+n}] \Vev{\left(\prod_{j=1}^n\sum_{m_j=1}^\infty \dfrac{J_{m_j} w_j^{m_j}}{m_j}\right) \cD^{-1}e^{\sum_{i=1}^\infty \frac{J_i t_i}{i\hbar}} \cD e^{\frac{J_{-2}}{2\hbar}}}^\circ.
	\end{align}
	Denote
	\begin{align}
		\Eop &\coloneqq  \cD^{-1}e^{\sum_{i=1}^\infty \frac{J_i t_i}{i\hbar}} \cD e^{\frac{J_{-2}}{2\hbar}},
	\end{align}
then, via the inclusion-exclusion formula, we can write	$\Mf^{\vee}_{g,n}$ as a finite sum of products of disconnected correlators:
\begin{align}
	&\Mf^{\vee
	}_{g,n}(w_1,\dots,w_n) =
	[\hbar^{2g-2+n}] \sum\limits_{l=1}^n\frac{1}{l}\sum_{\substack{I_1\sqcup\ldots\sqcup I_l=\llbracket n \rrbracket \\ \forall j\, I_j \neq \emptyset}}(-1)^{l-1} \prod_{i=1}^l  \Vev{\left(\prod_{j\in I_i}^n\sum_{m_j=1}^\infty \dfrac{J_{m_j} w_j^{m_j}}{m_j}\right)\Eop}.
\end{align}
Then, by the Leibniz rule, if we prove that for any $n'\in\mathbb{Z}_{>0},\,k'\in\mathbb{Z}_{\geq 0},\,b_1,\dots,b_{k'} \in\mathbb{Z}_{>0}$ any correlator of the form
\begin{equation}
	\Vev{\left(\prod_{j=1}^{n'}\sum_{m_j=1}^\infty \dfrac{J_{m_j} 
			P_{j,m_j}(z_j)
		}{m_j}\right)(\partial_{t_{b_1}}\cdots\partial_{t_{b_{k'}}}\Eop)|_{t=0}},
\end{equation}
where $P_{j,m_j}(z_j)$ are some polynomials in $z_j$, is a polynomial in $z_1,\dots,z_{n'}$, this would imply that $(\partial_{t_{a_1}}\cdots\partial_{t_{a_k}}\Mf^\vee_{g,n})|_{t=0}$ is a polynomial too. The appearance of the polynomials $P_{j,m_j}$ is due to the fact that, from~\eqref{eq:wmapsdef}, $w(z)$ is a polynomial in $z$ which also depends on the $t$-variables, and thus some of the $t$-derivatives after applying the Leibniz rule will act on the $w$'s, but they will only produce polynomials in the respective $z$-variables.

We have
\begin{align}\label{eq:evev}
	&\Vev{\left(\prod_{j=1}^{n'}\sum_{m_j=1}^\infty \dfrac{J_{m_j} 
			P_{j,m_j}(z_j)
		}{m_j}\right)(\partial_{t_{b_1}}\cdots\partial_{t_{b_{k'}}}\Eop)|_{t=0}}\\ \nonumber
	&= \Vev{\left(\prod_{j=1}^{n'}\sum_{m_j=1}^\infty \dfrac{J_{m_j} 
		P_{j,m_j}(z_j)}{m_j}\right)\cD^{-1}\left(\prod_{i=1}^{k'}\dfrac{J_{b_i}}{\hbar b_i}\right)\cD e^{\frac{J_{-2}}{2\hbar}}}\\ \nonumber
	&= \Vev{\left(\prod_{j=1}^{n'}\sum_{m_j=1}^\infty \dfrac{J_{m_j} 
	P_{j,m_j}(z_j)}{m_j}\right)\left(\prod_{i=1}^{k'}\cD^{-1}\dfrac{J_{b_i}}{\hbar b_i}\cD\right) e^{\frac{J_{-2}}{2\hbar}}}
\end{align}
From \cite[Proposition~3.1]{bychkov2021explicit}, since in our case $\phi(y)=1+y$ and recalling \eqref{eq:phimdef}, we have
\begin{align}
	&\cD^{-1}J_b\cD= \sum_{r=0}^\infty\partial_y^r\phi_a(y)\bigm|_{y=0}[u^r\zeta^b]
	\frac{
		e^{\sum_{i=1}^\infty u\,\hbar\cS(u\,\hbar\,i) J_{-i}\zeta^{-i}}
		e^{\sum_{i=1}^\infty u\,\hbar \cS(u\,\hbar\,i) J_{i}\zeta^{i}}}
	{u\,\hbar\cS(u\,\hbar)}\\ \nonumber
	& = \sum_{r=0}^b\left.\left(\partial_y^r\prod_{i=1}^b\left(1+y+\dfrac{2i-b-1}{2}\hbar\right)\right) \right|_{y=0}[u^r\zeta^b]
	\frac{
		e^{\sum_{i=1}^\infty u\,\hbar\cS(u\,\hbar\,i) J_{-i}\zeta^{-i}}
		e^{\sum_{i=1}^\infty u\,\hbar \cS(u\,\hbar\,i) J_{i}\zeta^{i}}}
	{u\,\hbar\cS(u\,\hbar)}.
\end{align}
The infinite sum in $r$ becomes finite in the second line since the $r$-th derivative acts on a polynomial of degree $a$. Thus $\cD^{-1}J_a\cD$ can be represented as a finite sum of expressions of the form
\begin{equation}
	R(\hbar) J_{-\alpha_1}\cdots J_{-\alpha_k} J_{\beta_1}\cdots J_{\beta_l},
\end{equation}
where $R(\hbar)$ is some polynomial in $\hbar$ and $\alpha_i,\beta_j \in \mathbb{Z}_{>0}$.

This implies that the whole expression~\eqref{eq:evev} is a finite sum of expressions of the form
\begin{align}\label{eq:jprodcor}
	\hbar^{-k}\widetilde{R}(\hbar)\; \VEV{\left(\prod_{j=1}^n\sum_{m_j=1}^\infty \dfrac{J_{m_j} 
			P_{j,m_j}(z_j)}{m_j}\right)\left(\prod_{i=1}^kJ_{-\alpha_{i,1}}\cdots J_{-\alpha_{i,k_i}} \cdot J_{\beta_{i,1}}\cdots J_{\beta_{i,l_i}}\right) \sum_{l=0}^\infty\dfrac{1}{l!}\left({\frac{J_{-2}}{2\hbar}}\right)^l},
\end{align}
where $\widetilde{R}(\hbar)$ is some polynomial in $\hbar$. But a VEV
\begin{equation}
	\Vev{J_{i_1}\cdots J_{i_k}}
\end{equation}
is not equal to zero if and only if the numbers $(i_1,\dots,i_k)$ perfectly split into pairs $(i,-i)$. Since there is only a finite number of $J_2$'s and a finite number of $J_{-j}$'s, $j\neq 2, j>0$, inside \eqref{eq:jprodcor}, only a finite number of terms one gets after expanding all brackets in \eqref{eq:jprodcor} are non-vanishing. Since all these terms are polynomial in $z_1,\dots,z_n$, the whole expression is polynomial, and thus all $(\partial_{t_{a_1}}\cdots\partial_{t_{a_k}}\Mf^\vee_{g,n})|_{t=0}$ are polynomial as well. This, as mentioned above, implies that all $\Mf^\vee_{g,n}(w(z_1),\dots,w(z_n))$ have no poles at $z_1=\pm 1$ or $z_1=0$.

\end{proof}
\begin{remark}
	Note that the argument in the proof of the preceding lemma does not work for poles at points with coordinates depending on $t$.
\end{remark}

\begin{lemma}\label{prop:noyp}
	$\Mf^\vee_{g,n}$ has no poles at $z_1$ such that $y(z_1)=-1$.
\end{lemma}
\begin{proof}
	This follows from Lemma~4.1 and the proof of Lemma~4.4 (formulas (149)--(152)) of \cite{bychkov2020topological}, up to a change of signs.
\end{proof}


\begin{lemma}\label{prop:nozip}
	For $2g-2+n>0$, $\Mf^\vee_{g,n}$ has no poles at $z_1=\infty$.	
\end{lemma}
\begin{proof}
 The proof literally repeats the proofs of \cite[Lemmata~4.6 and~4.10]{bychkov2020topological}.
\end{proof}

\begin{corollary}\label{cor:MTheta}
	For $2g-2+n>0$ we have $\Mf_{g,n}\in\Theta_n$.
\end{corollary}
\begin{proof}
	Follows from Theorem~\ref{prop:Mveeform} and Lemmas~\ref{prop:nopmone}, \ref{prop:noyp},  
	and \ref{prop:nozip}, taking into account that $\Mf_{g,n}$ is symmetric in all arguments.
\end{proof}

As we mentioned above, this Corollary implies the projection property, and thus concludes the proof of Theorem~\ref{th:BGF}.

\printbibliography
\end{document}